\documentclass[%
 reprint, superscriptaddress,
 amsmath,amssymb,
 aps,
]{revtex4-2}
\pdfoutput=1 

\usepackage{graphicx}
\usepackage[dvipsnames,table,xcdraw]{xcolor}
\usepackage{subcaption}
\usepackage{caption}
\usepackage{dcolumn}
\usepackage{bm}
\usepackage{bbold}
\usepackage{makecell}
\usepackage{tikz-cd}
\usepackage{soul}
\usepackage{float}
\usepackage{scalerel}
\usepackage{array}
\usepackage{booktabs}
\usepackage{gensymb}
\usepackage{physics}
\usepackage{amsmath,amsthm,amssymb,amsfonts,physics,siunitx,bbm,graphicx}
\usepackage{xfrac, hyperref}
\hypersetup{
    colorlinks,
    citecolor=blue,
    filecolor=blue,
    linkcolor=blue,
    urlcolor=blue
}
\usepackage{natbib}
\usepackage{mathtools}
\usepackage{cancel}
\usepackage{comment}
\usepackage{stmaryrd}
\usepackage{braket}
\usepackage{array, multirow}

\usepackage{svg}
\usepackage{subcaption}

\newtheorem{theorem}{Theorem}
\newtheorem{lemm}{Lemma}
\newtheorem{coro}{Corollary}
\newcommand{\ba}{\begin{eqnarray}}
\newcommand{\ea}{\end{eqnarray}}
\newcommand{\ban}{\begin{eqnarray*}}
\newcommand{\ean}{\end{eqnarray*}}
\newcommand{\one}{\mathbb{1}}

\newcommand{\cchn}{\varphi}
\newcommand{\chn}{\mathcal{F}}
\newcommand{\ite}{\bar} 

\newcommand{\rf}{\gamma}
\renewcommand{\vr}[1]{v^{#1}}
\newcommand{\mt}[1]{S^{#1}}
\newcommand{\gpaul}{\mathcal{P}}
\newcommand{\gtm}[1]{\mt{#1}_\gpaul} 
\newcommand{\ptm}[1]{\mt{#1}_\sigma} 

\newcommand{\cad}[1]{\mathsf{D_c}(#1)} 
\newcommand{\cfd}[1]{\mathsf{F_c}(#1)} 
\newcommand{\cfp}[1]{\mathsf{p^{#1}_c}} 
\newcommand{\skw}{\mathsf{S_c}} 

\newcommand{\qad}[1]{\mathsf{D_q}(#1)} 
\newcommand{\qfd}[1]{\mathsf{F_q}(#1)} 
\newcommand{\qfp}[1]{\mathsf{p^{#1}_{q}}} 

\newcommand{\cis}[1]{\mathrm{I}^\mathsf{s}_\mathsf{c}(#1)}
\newcommand{\qis}[1]{\mathrm{I}^\mathsf{s}_\mathsf{q}(#1)}

\newcommand{\ccd}[1]{\mathrm{I}^\mathsf{d}_\mathsf{c}(#1)}
\newcommand{\qcd}[1]{\mathrm{I}^\mathsf{d}_\mathsf{q}(#1)}

\newcommand{\gcm}{\mathrm{I}^{\mathsf{x}}_\mathsf{c}}

\newcommand{\spi}{\Xi_{3,2}}

\newcommand{\rlx}{\mathsf{t}_{\cchn,z}}

\newcommand{\rvp}{{\hat{\cchn}_{\rf}}}
\newcommand{\petz}{{\hat{\chn}_{\rf}}}
\newcommand{\tpose}{\text{\normalfont{T}}}
\newcommand{\qer}{\mathcal{W}}
\newcommand{\cer}{\Psi} 
\newcommand{\qabs}{\mathcal{A}}
\newcommand{\cabs}{\Upsilon}
\newcommand{\trnm}{T_\circlearrowleft} 
\newcommand{\tfrm}{T_{\scaleobj{0.72}{\leftarrow}}} 
\newcommand{\tfrr}{\vr{\tpose}_{\scaleobj{0.72}{\leftarrow}}}

\renewcommand{\rank}{\text{Rank}}
\renewcommand{\det}{\text{Det}}

\newcommand{\TrB}{\text{Tr}_{B}}
\renewcommand{\tr}[1]{\text{\normalfont{Tr}}\!\left[#1\right]}
\newcommand{\swap}{U_{\leftrightarrowtriangle}}

\allowdisplaybreaks

\begin{document}

\title{Quantifying Irreversibility via Bayesian Subjectivity \\ for Classical \& Quantum Linear Maps}%

\author{Lizhuo Liu}
\affiliation{Department of Physics, National University of Singapore, 2 Science Drive 3, Singapore 117542}

\author{Clive Cenxin Aw}
\affiliation{Centre for Quantum Technologies, National University of Singapore, 3 Science Drive 2, Singapore 117543}

\date{\today}

\begin{abstract}
In both classical and quantum physics, irreversible processes are described by maps that contract the space of states. The change in volume has often been taken as a natural quantifier of the amount of irreversibility. In Bayesian inference, loss of information results in the retrodiction for the initial state becoming increasingly influenced by the choice of reference prior. In this paper, we import this latter perspective into physics, by quantifying the irreversibility of any process with its Bayesian subjectivity---that is, the sensitivity of its retrodiction to one's prior. From this perspective, we review analytical and numerical results that highlight both intuitive and subtle insights that this measure sheds on irreversible processes. 
\end{abstract}

\maketitle

\section{Introduction}\label{intro}
Irreversibility lies at the heart of both physics and information theory. It manifests in all sorts of phenomena, from entropy production and inefficiencies, to irrecoverable errors and noise. In this work, we examine the irreversibility of processes through the lens of linear maps. While several features will be discussed in detail, we highlight one key aspect of linear maps here: irreversible maps are always \textit{contractive} (see Figure \ref{fig:deform} for illustrations)—their output domains form a strict subset of the input state space.

In a classical phase space, as the contrapositive of Liouville's theorem, changes in the volume containing a bundle of phase trajectories are a signature of dissipative dynamics \cite{marion2013classical-mechanics,goldstein2014classical-mechanics}. Some studies of irreversibility in the context of statistical mechanics make connections between irreversibility with deformations in state spaces as well \cite{hoover1994irr-to-vol,daems1999entropy-irr-to-vol,ramshaw2017entropy-irr-to-vol}. More recent explorations in information geometry also make similar observations \cite{gzyl2020geometry-entropy,wolfer2021information-geom-reverse-markov,van2021geometrical-info-irreversibility-markov,ito2018stochastic-thermo-info-geom,nicholson2018nonequilibrium-info-geom,kim2021information-geom-fluctuations-non-eq}. Such deformations are fundamental in the canonical derivation of fluctuation relations by Evans and Searles \cite{evans2002fluctuation}. This geometric picture is applicable to discrete spaces (simplexes) \cite{nielsen2020elementary-info-geom,amari2016information-text,pasieka2009geom-bloch-channel}, just as in quantum information, non-unitary evolutions are contractive \cite{gyongyosi2018survey-blochball-for-noisy,goyal2016geometry-of-bloch-ball-for-qutrits,pasieka2009geom-bloch-channel, rodriguez2022optimization-magnitude-determinant-unitary,greenbaum2015introductionTFRM}. 

With this in mind, this paper introduces another way to describe the irreversibility of any linear map. While irreversibility measures inspired by information theoretic notions like recoverability \cite{tajima2022universal-irr-measure,emori2023error-irr-measure} and divergences \cite{andrieux2024irreversibility-irr-measure,rong2018new-irr-measure,cisneros2023dissipative-irr-measure} have been fruitfully discussed, we turn our attention to specifically \textit{Bayesian} notions here. 

As noticed by Watanabe more than half a century ago, Bayes' rule can be used to define the \textit{reverse map} of any map, including irreversible ones, as soon as one understands reversal as retrodiction--- i.e. inference about the past \cite{watanabe55,watanabe65}. This intuition has been recently revived and extended to quantum information \cite{kwon-kim,BS21,AwBS}, with the Petz recovery map playing the role of quantum Bayes' rule \cite{petz,petz1,Leifer-Spekkens,QPRPetzPaper,petzisking2022axioms,aw2024-tabletop}. A well-known feature of Bayesian retrodiction in the presence of information loss is the necessity of a \textit{prior}. The Bayesian reverse is independent of the prior if and only if the map is reversible, or unitary in quantum theory \cite{AwBS}. At the other extreme, the Bayesian reverse of a map that erases all information about the input (``erasure map'') is another erasure map, which sends the whole state space to the prior--indeed, there is no reason to update one's belief if no information is added. 

This leads to the central observation of this paper: \textit{the degree of irreversibility of a map can be quantified by the dependence of its Bayesian inverse on the prior}. In other words, we characterize the irreversibility of a process by the \textit{subjectivity} inherent in performing retrodiction on it. This measure, which we call \textit{Bayesian subjectivity}, yields a range of insights and implications supported by analytical and numerical results. In particular, the measure yields physical intuitions about irreversibility native to thermodynamics (such as notions of quasistaticity and entropy production), satisfies the data processing inequality in numerical simulations and detects subtler sources of reversibility. We discuss these and other features in the proceeding sections. 

We begin by reviewing the formal framework of linear maps in Section \ref{sec:formal-maps}, along with the Bayesian inference tools that will later play a central role. Section \ref{sec:illus-irrev} introduces key examples of Bayesian inversion that illustrate the concept of Bayesian subjectivity—the centrepiece of this work, which is introduced proper in Section \ref{sec:m-irr}, where general theorems and results are also discussed. In Section \ref{sec:resul-discuss}, we present numerical results for specific dimensionalities, and highlight their main features. Finally, Section \ref{sec:concl} summarizes our findings and concludes the study.

\begin{center}
    \begin{figure*}[ht]
    \centering
    \includegraphics[width=1\textwidth]{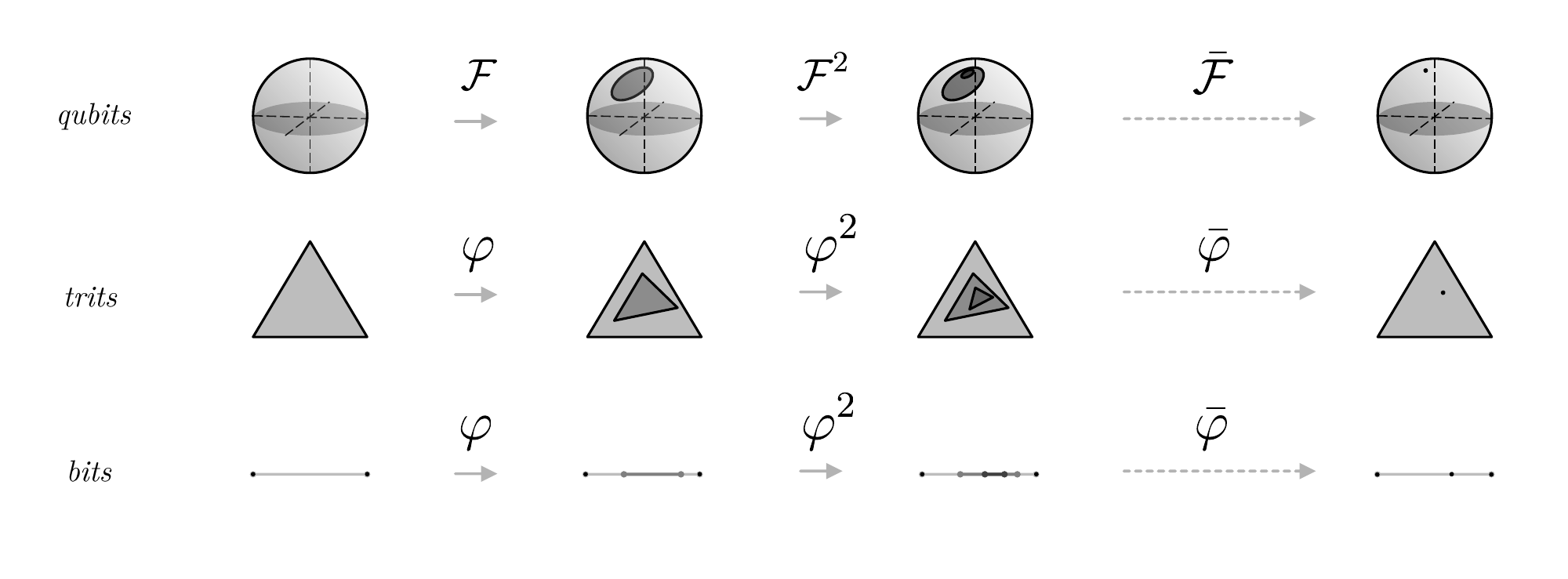}
    \caption{An illustration of the information geometric action of various, generic maps on their respective state spaces. Namely, for generic maps acting on qubits, trits and bits. $\ite\chn$ and $\ite\cchn$ are the respective channels iterated over an arbitrarily long time \eqref{eq:ite}.}
    \label{fig:deform}
\end{figure*}
\end{center} 

\section{Formalism for Classical \& Quantum Processes}\label{sec:formal-maps}
\subsection{Classical Formalism}
In the classical regime, we work with discrete state spaces of dimension $d$. The description of a physical system is thus represented by a probability distribution $p$: $0 \leq p(a)\leq 1$ for $a \in \{1,...,d$\}; and $\sum_{a=1}^{d} p(a)=1$. It can be represented by a $d-$dimensioned probability vector $\vr{p}$, whose entries are $\vr{p}_a:= p(a)$. 

A classical process can be represented as a stochastic map $\cchn$, defined by $d \times d'$ conditional probabilities $\cchn(a'|a)$ transiting from the input state $a$ to the output state $a' \in \{1,...,d'\}$. These probabilities must satisfy $0\leq\cchn(a'|a)\leq 1$ for all $a,a'$ and $\sum_{a'=1}^{d'}\cchn(a'|a)=1$ for all $a$. These channels can be represented by column-stochastic matrix $\mt{\cchn}$ with entries  $\mt{\cchn}_{a'a}:=\cchn(a'|a)$. The output of the state $p$ through $\cchn$ is thus given by \ba  \mt{\cchn}\vr{p} &=& \vr{\cchn[p]}, \\ \label{eq:c-posterior}
 \cchn[p](a') &=& \sum^{d}_{a=1} \cchn(a'|a) p(a). \ea Without loss of generality (due to the freedom afforded by adding redundancies), we set $d' = d$. 
 \subsubsection*{Bayesian Inversion for Classical Maps}
 Bayes' rule dictates that the forward channel $\cchn$ is insufficient in making an inference on the input that produced some given observed output $q$. This output, in general, requires a reverse or retrodiction map defined by $\cchn$ some prior reference $\rf$ \cite{AwBS, watanabe65}. Formally speaking, given $q$, our updated guess on the input is given by $\rvp[q]$, where $\rvp$ is the map corresponding to retrodiction: 
 \begin{equation} \label{bayes}
    \rvp(a|a')=\cchn(a'|a)\frac{\rf(a)}{\cchn[\rf](a')}.
\end{equation}
One primary interest of this work is dependence of $\rvp$ on $\rf$, which will later be formalized as the \textit{subjectivity} inherent to retrodiction on $\cchn$.

\begin{center}
\begin{figure*}[ht]
    \centering
    \includegraphics[width=0.91\textwidth]{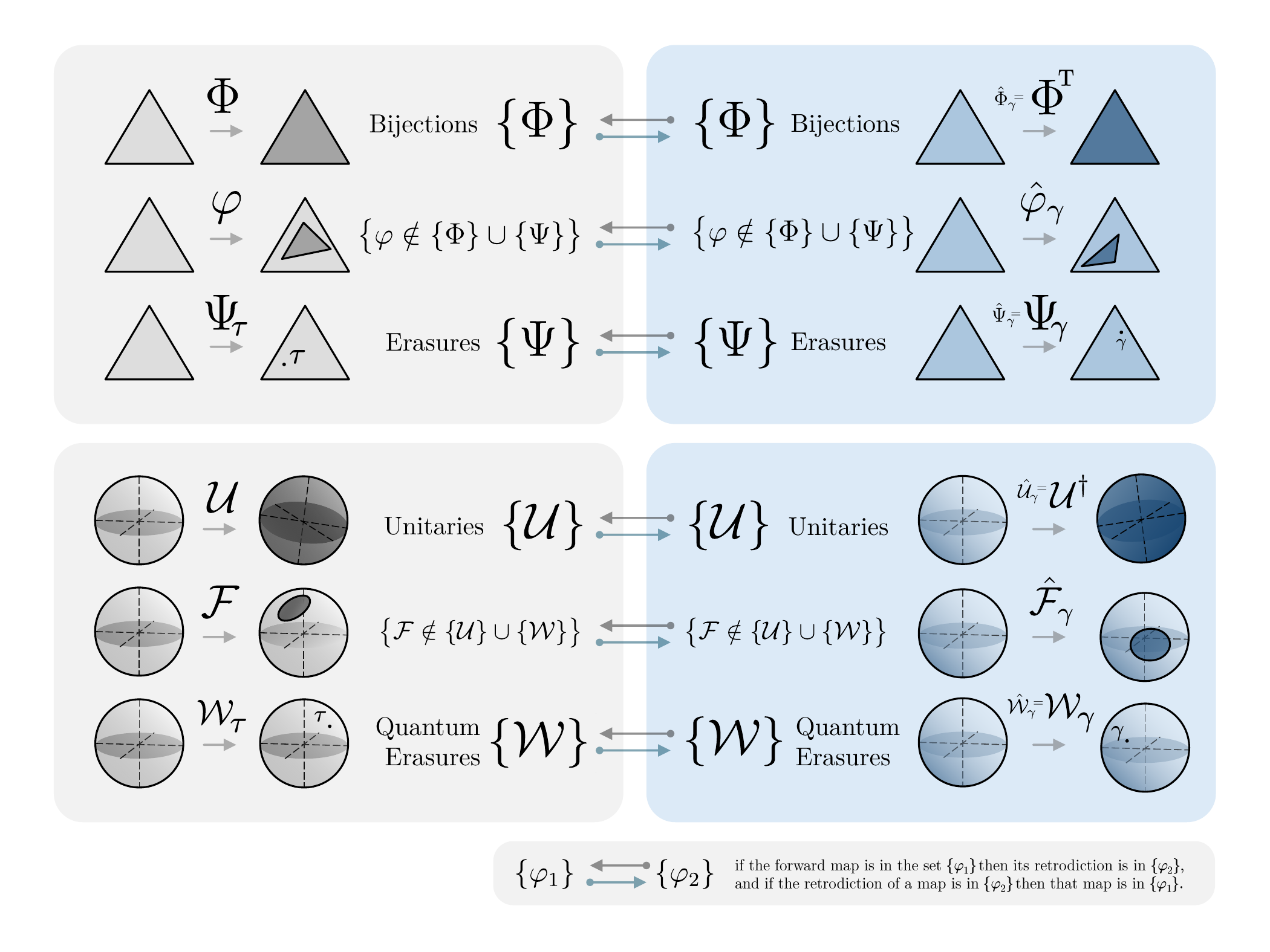}
    \caption{An illustration of various edge cases of reversibility and irreversibility, for the quantum and classical maps, and their respective Bayesian retrodictions.}
    \label{fig:retrogeom}
\end{figure*}
\end{center}

\subsection{Quantum Formalism}
We model quantum processes under a channel-theoretic framework through completely positive, trace-preserving maps $\chn$ working on semidefinite operators $\rho \succeq 0$ in a $d$-dimensioned Hilbert space $\mathbb{C}^d$ with $\Tr[\rho]=1$. While these channels certainly can send states from one Hilbert space to another with different dimensionality, we will focus on representations that map states to the same discrete Hilbert space $\mathbb{C}^{d} \to \mathbb{C}^{d}$. Every such map can be written under a unitary dilation representation. Any quantum channel can be seen as the marginal of a global unitary $U \in \mathbb{C}^{d \times d_B}$ acting on a target input in $\mathbb{C}^{d}$ and an ancillary system $\beta$ in $\mathbb{C}^{d_B}$ \cite{wilde_2013}:
\begin{equation}\label{eq:chndil}
    \chn[\bullet] = \TrB \left[U (\bullet \otimes \beta_B) U^\dagger\right] 
\end{equation}
Alternatively, it may be written in a Kraus form: 
\begin{equation}\label{eq:chnkraus}
    \chn[\bullet] = \sum_{i \in \{\kappa_i\}} \kappa_i \bullet \kappa_i^\dagger.
\end{equation}
where $\sum_i \kappa_i^\dagger \kappa_i = \one$. 
\subsubsection*{Bayesian Inversion on Quantum Channels}
While there are certainly other choices of recovery maps for this purpose \cite{petzisking2022axioms}, we take the Petz recovery map as the quantum analogue for Bayesian inversion for transformations in this regime, doing so largely on the basis of axiomatic and conceptual reasons \cite{petzisking2022axioms, QPRPetzPaper, BS21, aw2024-tabletop, Leifer-Spekkens}. As with Bayes' rule, the forward channel $\chn$ is insufficient for inferring what the input on that channel is given some observed output. The Petz recovery map $\petz$ requires also a quantum-theoretic reference state $\gamma$, a density operator that plays the role of a Bayesian prior:
\begin{eqnarray}\label{petz}
    \petz[\bullet] = \sqrt{\rf} \, \chn^\dagger\left[\frac{1}{\sqrt{\chn[\rf]}} \bullet \frac{1}{\sqrt{\chn[\rf]}}\right] \sqrt{\rf},
\end{eqnarray}
where $\chn^\dagger$ is the adjoint of $\chn$ fulfilling for all self-adjoint $X,Y$ the relation \begin{equation}\label{eq:adj}
    \Tr\Big[\chn[X] Y\Big] = \Tr\left[X\chn^\dagger[Y]\right].
\end{equation} 
\newline
For the dilation picture this is given by
\begin{equation}\label{eq:adjdil}
    \chn^\dag[\bullet] = \TrB\left[\sqrt{\one \otimes \beta} U^\dag (\bullet \otimes \one) U \sqrt{\one \otimes \beta}\right],
\end{equation}
while in the Kraus representation this is
\begin{equation}\label{eq:krausdil}
\chn^\dag[\bullet] = \sum_{i \in \{\kappa_i\}} \kappa_i^\dag \bullet \kappa_i.
\end{equation}
Thus, one's update on a quantum-theoretic postulate $\rf$, given that we know an observed output $\eta$ underwent an effective transformation $\chn$, is given by $\petz[\eta]$.

\section{Extremal Maps \& Corresponding Bayesian Inversions} \label{sec:illus-irrev}
 We move now to discuss two notable classes of linear maps and their respective retrodiction properties, which will then lead naturally to our central research question. 
 
 For these families of maps, we note in particular their \textit{absolute determinant}, which quantifies the extent to which the state space is preserved by the map. For the classical case, this is naturally defined via matrix representation: 
 \begin{equation}\label{eq:cad}
    \cad{\cchn} = |\det(\mt{\cchn})|
\end{equation}
 The quantum counterpart of this quantifier, denoted by $\qad\chn$, requires a little bit more involvement but is conceptually the same. $\qad\chn=1$ if and only if the input state space is totally preserved by $\chn$ and $\qad\chn=0$ if and only if the $\chn$ is singular. To avoid encumbering the main text, the definition \eqref{qad} and technicalities are organized into Appendix \ref{sec:geom}. 

\subsection{Bijective / Unitary Channels}\label{sec:bijunichn}
It is natural to first look at the maps for which the absolute determinant is unity. In the phase space picture, these are dynamics that obey Liouville's theorem and preserve phase volumes.

For classical stochastic maps these correspond to bijections, which we denote by $\Phi$. Bijections can be seen as relabellings, where index goes ``one-to-one and one-from-one''. These may be represented as permutation matrices $\mt{\Phi}$---bistochastic matrices for which every row and every column has a single $1$-entry. The inverse and transpose of matrices coincide i.e.~$(\mt{\Phi})^{-1}=(\mt{\Phi})^{\tpose}$. It is easy to see that \footnote{$\mt\Phi$ are full-rank matrices with eigenvalues of $1$ for the whole eigenbasis. So clearly $\cad{\cchn} = 1$. Meanwhile, since stochastic matrices have a spectral radius of one, the only maps for which the absolute determinant will be conserved would be those that are bijections. },
\begin{lemm} \label{thm:bij-abs-det} A classical stochastic map conserves the state space if and only if it is bijective:
\begin{equation}
    \cad{\cchn} = 1 \quad \Leftrightarrow \quad \cchn \text{ is a bijection}
\end{equation}
\end{lemm}
In the quantum regime, the reverse process for any unitary $\mathcal{U}$ is also its inverse $\mathcal{U}^{-1}=\mathcal{U}^\dag$. Similarly \footnote{Since every choice of operator basis $\{\gpaul_k\}$ is equivalent up to a unitary transformation, then by \eqref{eq:gtm-def}, $\gtm{\mathcal{U}}$, which is the same for all representations after normalization, will always have absolute determinant $1$. Meanwhile, if $\qad{\chn} = 1$, then the whole state space is preserved and $\chn$ is simply a rotation of space. Thus, $\chn$ is a change of basis---a unitary transformation.},
\begin{lemm}  \label{thm:quni-abs-det} A quantum channel conserves the state space if and only if it is unitary:
\begin{equation}
    \qad{\chn} = 1 \quad \Leftrightarrow \quad \chn \, \text{ is unitary}.
\end{equation}
\end{lemm}

As argued more thoroughly in \cite{AwBS} and \cite{petzisking2022axioms}, it is known that applying \eqref{bayes} and \eqref{petz} on bijections and unitary channels respectively gives some notable insights. For classical maps,
\begin{lemm}\label{thm1}
These three statements are equivalent: 
\begin{enumerate}
    \item[\emph{(\textbf{I})}] $\cchn^\tpose\cchn = \one$ i.e. the channel is a bijection. 
    \item[\emph{(\textbf{II})}] $ \forall \rf : \rvp = \hat{\cchn}$ i.e. the channel's Bayesian inverse is independent of the reference prior. 
    \item[\emph{(\textbf{III})}] $\exists\rf:\rvp = \cchn^{-1}$ i.e. there exists a reference prior, for which the channel's Bayesian inverse is its matrix inverse. 
\end{enumerate}
\end{lemm} 
Unitary channels have similar relationships to their Petz transpose:
\begin{lemm}\label{thm2}
These three statements are equivalent:
\begin{enumerate}
    \item[\emph{(\textbf{I})}]  $\chn\circ \chn^{\dagger} = \one$ i.e. the channel is unitary.
    \item[\emph{(\textbf{II})}] $\forall\rf : \petz =\hat{\chn}$ i.e. the channel's Petz recovery map is independent of the reference state.
    \item[\emph{(\textbf{III})}] $\exists \rf : \petz ={\chn}^{-1}$ i.e. there exists a reference state, for which the channel's Petz recovery map is its inverse channel $\chn^{-1}$.
\end{enumerate}
\end{lemm} 

For both Lemmas, we emphasize here the equivalence of (\textbf{I}) and (\textbf{II}). Total reversibility is identified with independence from one's prior when performing Bayesian inversion.

\subsection{Erasure Channels}\label{sec:er-chns}
While bijective channels are reversible and preserve the whole space of states, we can consider the opposite: channels for which the entire space of states reduces to a point. We may call these \textit{erasure} channels (otherwise known as \textit{discard-and-prepare} channels), they represent the extreme of information geometric irreversibility. We define classical erasure channels $\cer$ in the following way: 
\begin{eqnarray}
    \forall p \,\exists\tau: \; \cchn[p] = \tau  \;  &\Leftrightarrow& \;   \cchn \text{ is a classical erasure},
\end{eqnarray}
where $\tau$ is some probability distribution which the entire space of distributions is erased to. On the level of individual probability transitions, 
\begin{equation}
    \label{eq:cer-indiv} \forall (a,a'): \cer(a'|a)=\tau(a'), 
\end{equation} which implies, from \eqref{eq:c-posterior},  that all $\cer[p](a')=\sum_a p(a)\tau(a')=\tau(a')$. Thus, $\cchn$ is a classical erasure if and only if $\rank{(\mt{\cchn})}=1$ \footnote{From \eqref{eq:cer-indiv}, $\mt{\cer}$ have $\vr{\tau}$ as all of their columns, which in turn implies that the object is rank-$1$ and has zero absolute determinant. The opposite implication is also easily verified. All rank-$1$ matrices can be written as an outer product of two vectors: $\mt{M}=\ketbra{u}{v} \Rightarrow \mt{M}_{ij} = u_i v_j$. Since we are dealing with column-stochastic rank-$1$ matrices, then $\forall j : \sum_i \mt{\cchn}_{ij}=1 \Rightarrow \forall j : \mt{\cchn}_{ij} = {u_i}/{\sum_i u_i}$, which is to restate \eqref{eq:cer-indiv}.}. 

Now, quantum erasure channels $\qer$ may be written as:
\begin{equation}
    \forall \rho : \; \chn[\rho] = \tau \;\; \Leftrightarrow \;\;  \chn \text{ is a quantum erasure,}
\end{equation}
where $\tau$ here is a density operator, of which the entire space of states is erased to. It can be realized by the dilation $\qer[\rho] = \TrB[\swap[\rho \otimes \tau]\swap^\dag]$, with $\swap$ as the swap operator with the unitary action $\swap \ket{\psi} \otimes \ket{\phi} = \ket{\phi} \otimes \ket{\psi}$ for all $\ket{\psi}, \ket{\phi} \in \mathbb{C}$ \cite{scarani-ziman-2001quantum-homogenization,scarani2002thermalizing}. Now, when $\gtm{\chn}$ is some linear representation of the quantum channel $\chn$ \eqref{eq:gtm-def}, $\chn$ is a quantum erasure if and only if $\rank{(\gtm{\chn})}=1$ \footnote{Noting \eqref{eq:gtm-def}, $\forall (j \neq d^2): \qer[\gpaul_j] = \Tr[\gpaul_j] \tau=\mathbb{0}$. Hence for these values of $\forall (j\neq 1 \vee i =d^2)$, $(\gtm{\qer})_{ij} = 0$. Meanwhile, for any other entry,  $(\gtm{\qer})_{ij} = \Tr[\gpaul_i \tau]$---leaving only one non-zero column. This ensures that $\rank(\gtm{\qer})=1$. Likewise, if $\rank(\gtm{\chn})=1$, $(\gtm{\chn})_{ij} = u_i v_j$. Now, it is always the case that $(\gtm{\chn})_{d^2d^2}=1$ and for all other $j$, $(\gtm{\chn})_{d^2j}=0$, $v_j$ is $0$ for all $j$ but $d^2$. Which also implies that all $i\neq d^2 \wedge j\neq d^2$ $(\gtm{\chn})_{ij}=0$, which is to say the state space is totally erased.}. 

Finally, we remark that erasure channels always have zero absolute determinant ($\cad\cer =0, \qad\qer=0$) \footnote{It is however, not the case that any map with a zero absolute determinant is an erasure. For higher dimensions, we may have $d$-dimensional channels for which their corresponding matrices are rank-$n$ for $1<n<d$. Such channels are not erasing over the whole state space but nevertheless have an absolute determinant of zero, as they erase a \textit{subspace}. We may refer to such maps as \textit{partial erasure} maps. We note here that, in the case of \textit{bits}, there are no partial erasure channels. For this elementary class of stochastic maps, it is the case that $\cad{\cchn}=0$ if and only if $\cchn$ is a classical erasure.}.

We now turn to how Bayes' rule treats erasure channels. For the classical erasure, we can prove the following equivalences, 

\begin{lemm}
\label{res:c-erase}
These three statements are equivalent: 
\begin{enumerate}
    \item[\emph{(\textbf{I})}] $ \forall p\,\exists\tau : \, \cchn[p] = \tau$ i.e. the channel is a classical erasure. 
    \item[\emph{(\textbf{II})}] $ \forall\rf\,\forall p : \, \rvp[p] = \rf$ i.e. the channel's Bayesian inverses are always classical erasures that erase toward the reference prior.
    \item[\emph{(\textbf{III})}] $ \exists\rf\, \exists\mu \, \forall p: \, \rvp[\rho] = \mu$ i.e. there exists some reference prior, for which the channel's Bayesian inverse is a classical erasure. 
\end{enumerate}
\end{lemm} 
\begin{proof}
The relationship (\textbf{I}) $\to$ (\textbf{II}) can be easily shown:.
\begin{eqnarray*}
    \hat{\cer}_\rf(a|a') &=& \cer(a'|a)\frac{\rf(a)}{\cer[\rf](a')} = \tau(a')\frac{\rf(a)}{\tau(a')} = \rf(a),
\end{eqnarray*}
which is yet another erasure channel \eqref{eq:cer-indiv}. Specifically, it is a channel that \textit{erases to the reference prior}. This is sensible since all information is lost in the channel (there is nothing to be learnt from action of the channel), Bayes' rule defers to the reference prior. (\textbf{II}) $\to$ (\textbf{I}) holds because the retrodiction erases to the prior \textit{for any prior}:
\begin{align*}
\forall(a',a,\rf) \quad & \rvp(a|a') = \rf(a) \\
&\cchn(a'|a)\frac{\rf(a)}{\cchn[\rf](a')} = \rf(a) \\
& \cchn(a'|a) = \cchn[\rf](a') \quad \text{noting $\forall\rf$} \\
\Rightarrow \qquad &  \cchn(a'|a) = \tau(a') 
\end{align*}
Now, (\textbf{II}) $\to$ (\textbf{III}) holds by mere instantiation. Finally, (\textbf{III}) $\to$ (\textbf{II}) holds because of the recoverability condition that Bayesian inversions always fulfill \cite{petzisking2022axioms}: $\forall\rf:\rvp\circ\cchn[\rf]=\rf$. Together with (\textbf{III}), since $\cchn[\rf]$ is an instantiation of $p$, then the fixed point $\mu$ must be $\rf$.
\end{proof}

\begin{center}
\begin{figure*}[t]
    \centering
    \includegraphics[width=0.82\linewidth]{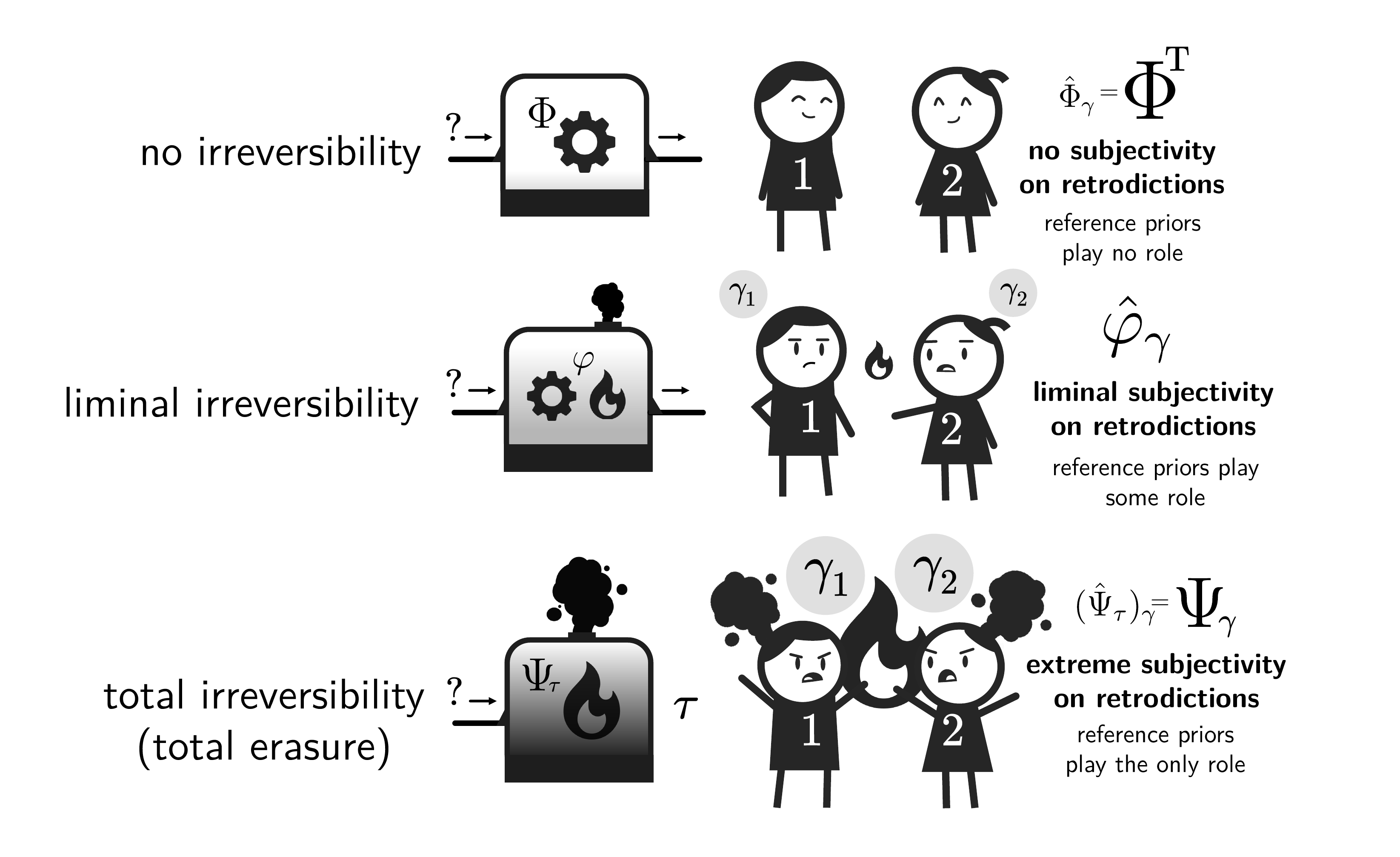}
    \caption{A cartoon illustration of the key interest in this work: quantifying the irreversibility of a map by a formal measure of the dependence of its Bayesian inversion on reference priors. That is, understanding physical irreversibility of processes through the subjectivity involved in doing Bayesian inference for their past inputs.}
    \label{fig:subjectivity}
\end{figure*}
\end{center}

When the Petz map is applied to the quantum erasure channel, we have similar set of statements:
\begin{lemm} \label{res:q-erase}
These three statements are equivalent: 
\begin{enumerate}
    \item[\emph{(\textbf{I})}] $ \forall \rho\,\exists\tau : \, \chn[\rho] = \tau$ i.e. the channel is a quantum erasure. 
    \item[\emph{(\textbf{II})}] $ \forall\rf\,\forall \rho : \, \petz[\rho] = \rf$ i.e. the channel's Petz transposes are always quantum erasures that erase toward the reference state.
    \item[\emph{(\textbf{III})}] $ \exists\rf\, \exists\mu \, \forall \rho: \, \petz[\rho] = \mu$ i.e. there exists some reference state, for which the channel's Petz transpose is a quantum erasure. 
\end{enumerate}
\end{lemm} 
The proof is similar to Lemma \ref{res:c-erase}'s (see Appendix \ref{app-qer}). 

\subsection{State Space Preservation \& Dependence \\ on Bayesian Priors}
These Lemmas are raised to make the following equivalence relations between information geometry and the more Bayesian notion of dependence on priors. From the Lemmas in Section \ref{sec:bijunichn} one easily derives that,
\begin{theorem}  \label{thm:adet-ref-rls}
Any map conserves the state space completely if and only if its Bayesian inversion is always independent of the reference prior. 
\begin{eqnarray*} 
    \cad{\cchn} = 1 \quad &\Leftrightarrow& \quad \forall \rf : \rvp = \hat{\cchn} \\
    \qad{\chn} = 1 \quad &\Leftrightarrow& \quad \forall \rf : \petz = \hat{\chn}
\end{eqnarray*}
\end{theorem} 

When maps conserve all information, there is no Bayesian subjectivity for their retrodiction. Meanwhile, whenever the state space shrinks, such subjectivity \textit{must} enter in. Meanwhile, from the Lemmas in \ref{sec:er-chns}, we easily derive a counterpart, 

\begin{theorem}\label{thm:adet-er-ref-rls}
    Any map erases an entire state space to a single point (that is, a map is rank-$1$) if and only if its Bayesian inversion erases to the reference prior. 
    \begin{eqnarray}
     \normalfont{\rank}(\mt \cchn)=1 \quad &\Leftrightarrow& \quad   \forall\rf\,\forall p : \, \rvp[p] = \rf \\
     \normalfont{\rank}(\gtm \chn)=1  \quad &\Leftrightarrow& \quad \forall\rf\,\forall \rho : \, \petz[\rho] = \rf
    \end{eqnarray}
\end{theorem}
These two theorems strongly connect information geometric intuitions with the Bayesian dependence on one’s choice of prior, at the limiting cases of reversibility and irreversibility. When a map \textit{preserves} the whole space, its retrodiction has \textit{no dependence on our prior}. Meanwhile, when a map \textit{erases} the whole space, its retrodiction \textit{depends on nothing but our prior}, so to speak. 

A natural question then arises: can this relationship be formalized through a measure of irreversibility that quantifies such prior dependence for arbitrary linear maps? This leads us to this work's central construct.

\section{Bayesian Subjectivity} 
\label{sec:m-irr}
We quantify the dependence of a map's Bayesian inversion on its prior, as an irreversibility measure, through the \textit{Bayesian Subjectivity} $\mathrm{I^\mathsf s}$ of that map. For the classical case, we formalize this as:
\begin{equation}\label{eq:cis-def}
    \cis\cchn= \int || \hat\cchn_{\rf_1} -\hat\cchn_{\rf_2} ||_\lambda \; d\rf_1 d\rf_2,
\end{equation}
where $||\cchn_1 -\cchn_2||_\lambda = \sqrt{\lambda_{\text{max}}\big((\mt{\cchn_1}-\mt{\cchn_2})^\tpose(\mt{\cchn_1}-\mt{\cchn_2})\big)}$ is largest singular value of the matrix $\mt{\cchn_1}-\mt{\cchn_2}$ (with $\lambda_\text{max}(S)$ denoting the largest eigenvalue of some matrix $S$), which is simply a distance measure between the matrices $\mt{\cchn_1}$ and $\mt{\cchn_2}$. In this way, $|| \hat\cchn_{\rf_1} -\hat\cchn_{\rf_2} ||_\lambda$ formalizes the retrodictive ``disagreement'' that emerges from the difference between the two postulated priors $\rf_1, \rf_2$. $\cis\cchn$ then is summation over all such disagreement elements, over across all pairs of all possible priors. Through this construction, the measure captures the sensitivity of map's retrodiction to the choice of priors. We illustrate these notions through a cartoon in Figure \ref{fig:subjectivity} \footnote{Gear, flame and smoke icons included in Figure \ref{fig:subjectivity} are made by \textit{Freepik} and \textit{C-mo Box} from \textit{www.flaticon.com}, under the Flaticon attribution license.}. 

In the quantum case, we have the corresponding expression for Bayesian subjectivity:
\begin{equation}\label{eq:qis-def}
    \qis\chn = \int || \hat\chn_{\rf_1} -\hat\chn_{\rf_2} ||_\diamond \; d\rf_1 d\rf_2,
\end{equation}
where, $||\chn_1 -\chn_2||_\diamond$ is the diamond norm distance between the two quantum channels $\chn_1$ and $\chn_2$. As with \eqref{eq:cis-def}, $|| \hat\chn_{\rf_1} -\hat\chn_{\rf_2} ||_\diamond$ then describes the corresponding quantum-theoretic discrepancies (the quantum ``disagreement'') between two Bayesian inversions (given by the Petz map) arising from different quantum-theoretic reference priors. 

\subsection{Bayesian Subjectivity for Extremal Maps}\label{ssec:irr-on-extremes}
With this, Lemmas \ref{thm1} and \ref{thm2} clearly imply that every disagreement element is zero for bijective and unitary maps. Hence we have, 
\begin{theorem} \label{thm:bij-abs-det-irrv-measure} Bayesian subjectivity is zero if and only if the map is bijective (or unitary),
\begin{eqnarray*}
    \cis\cchn=0 \; &\Leftrightarrow& \; \cchn \text{ is a bjiection.} \\
   \qis\chn=0 \; &\Leftrightarrow& \; \chn \text{ is a unitary.} 
\end{eqnarray*}
\end{theorem}

Likewise, since every erasure channel's retrodiction is always an erasure channel \textit{to the prior} (regardless of the original map's fixed point), Lemmas \ref{res:c-erase} and \ref{res:q-erase} give the following:
\begin{theorem} \label{thm:erasure-irr-measure-all-equal}
For any given choice of dimension, erasures for the state space of that dimension always share the same values of Bayesian subjectivity:
\begin{eqnarray*}
    \forall(\cer_1,\cer_2)\in \mathbb{R}^d &:&  \cis{\cer_1} = \cis{\cer_2} \\
   \forall(\qer_1,\qer_2)\in \mathbb{C}^d &:&  \qis{\qer_1} = \qis{\qer_2} 
\end{eqnarray*}
\end{theorem}
These are instantiated in the numerics later.

Since Theorem \ref{thm:erasure-irr-measure-all-equal} holds, we normalize the measure for each dimensional context: for a given $d$, for all $\cchn,\cer \in \mathbb{R}^d$ and $\chn,\qer \in \mathbb C^d$, 
\begin{eqnarray}\label{eq:normalization-irrev-measures}
     \frac{\mathrm{I}^{\mathsf{s}}_\mathsf{c}(\cchn)}{\mathrm{I}^{\mathsf{s}}_\mathsf{c}(\cer)} \mapsto {\mathrm{I}^{\mathsf{s}}_\mathsf{c}(\cchn)}, \quad
     \frac{\mathrm{I}^{\mathsf{s}}_\mathsf{q}(\chn)}{\mathrm{I}^{\mathsf{s}}_\mathsf{q}(\qer)} \mapsto {\mathrm{I}^{\mathsf{s}}_\mathsf{q}(\chn)}.
\end{eqnarray}
This gives a clearer yardstick of irreversibility for any $d$-dimension. We also conjecture that, for any choice of $d$, erasures channels maximizes subjectivity. If this holds, then $\cis\cchn, \,  \qis\chn \in[0,1]$ always, with bijections and erasures at the lower and upper bounds respectively.

\subsection{Concatenations \& Data Processing Inequality}
Sound irreversibility measures ought to obey some expression of the data processing inequality: \textit{Markovian concatenations should never decrease a measure of irreversibility}. Now, we have numerically verified that this indeed holds for all cases we have sampled thus far (order of $10^3$ samples for arbitrary dimensions and channels). Specifically, we have numerically verified that Markovian concatenations both classical ($\psi \,\circ\,\cchn$) and quantum ($\mathcal{G}\circ\chn$) obey the following:
\begin{eqnarray} 
    \mathrm{I}_\mathsf{c}^\mathsf{s}(\psi\circ\cchn) \; &\geq& \; \mathrm{I}_\mathsf{c}^\mathsf{s}(\cchn) \label{eq:c-dpi} \\
    \mathrm{I}_\mathsf{q}^\mathsf{s}(\mathcal{G}\circ\chn) \; &\geq& \;  \mathrm{I}_\mathsf{q}^\mathsf{s}(\chn) \label{eq:q-dpi}
\end{eqnarray}
Analytical proof for this, however, remains elusive. For instance, it is not possible to prove this by showing that \textit{disagreement elements} never decrease over concatenations. There are tuples of priors and channels that \textit{decrease} disagreement over concatenations. Formally, $\exists(\cchn, \Omega=\psi\circ \cchn, \rf_1, \rf_2):$
\begin{equation}
|| \hat\Omega_{\rf_1} -\hat\Omega_{\rf_2} ||_\lambda
<
|| \hat\cchn_{\rf_1} -\hat\cchn_{\rf_2} ||_\lambda 
\end{equation}
This also occurs in quantum scenarios. That said, we emphasize that while this occurs for some pair of priors (that is, \textit{element-wise}), integration over \textit{all} priors afterward prevents subjectivity as a whole never decreases (that is, on the level of the \textit{measure}). Therefore, Bayesian subjectivity's adherence to the data processing inequality is an \textit{aggregate} property. 

Now, from the composability of Bayesian inverses \cite{petzisking2022axioms},
\begin{eqnarray}
    \forall(\cchn,\Omega=\psi\circ \cchn) &:& \hat{\Omega}_\rf=\rvp \circ \hat\psi_{\cchn[\rf]} \\
    \forall(\chn,\mathcal{T}=\mathcal{G}\circ \chn) &:& \hat{\mathcal{T}}_\rf=\petz\circ \hat{\mathcal{G}}_{\chn[\rf]}
\end{eqnarray}
alongside Lemmas \ref{thm1}, \ref{thm2} and Theorem \ref{thm:erasure-irr-measure-all-equal}, we have:
\begin{theorem}\label{thm:concat-preserve}
Bijections (or unitaries) always preserve Bayesian subjectivity, and the Bayesian subjectivity of an erasure channel can never be increased. 
    \begin{eqnarray}
    \forall(\Phi,\cchn) &:&  \cis{\Phi\circ\cchn} = \cis{\cchn} \label{rln1} \\
    \forall(\mathcal{U},\chn) &:& \qis{\mathcal{U}\circ\chn} = \qis{\chn} \label{rln2} \\
    \forall(\cer,\cchn) &:&  \cis{\cchn\circ\cer} = \cis{\cer} \label{rln3} \\
    \forall(\qer,\chn) &:&  \qis{\chn\circ\qer} =\qis{\qer} \label{rln4}
    \end{eqnarray}
\end{theorem}
For detailed proofs of these relations, see Appendix \ref{app-proof-concat-cis}. From these we get the expected features of reversible process always making no difference to any original irreversibility value, and the irreversibility value of any maximally irreversible process being impossible to increase by the action of any other process.

\section{Results for Bits, Qubits and Trits} \label{sec:resul-discuss}
In the previous section, we have reviewed \textit{general} results that show that Bayesian subjectivity constitutes an insightful and viable irreversibility measure. We now move to discussing more specific dimensionalities. We do this for bit channels $\cchn\in\mathbb{R}^2$, qubit channels $\chn\in\mathbb{C}^2$ and finally trit channels $\cchn\in\mathbb{R}^3$.

\begin{figure}[h!] 
\begin{subfigure}{0.45\textwidth}
\centering
         \includegraphics[width=0.94\linewidth]{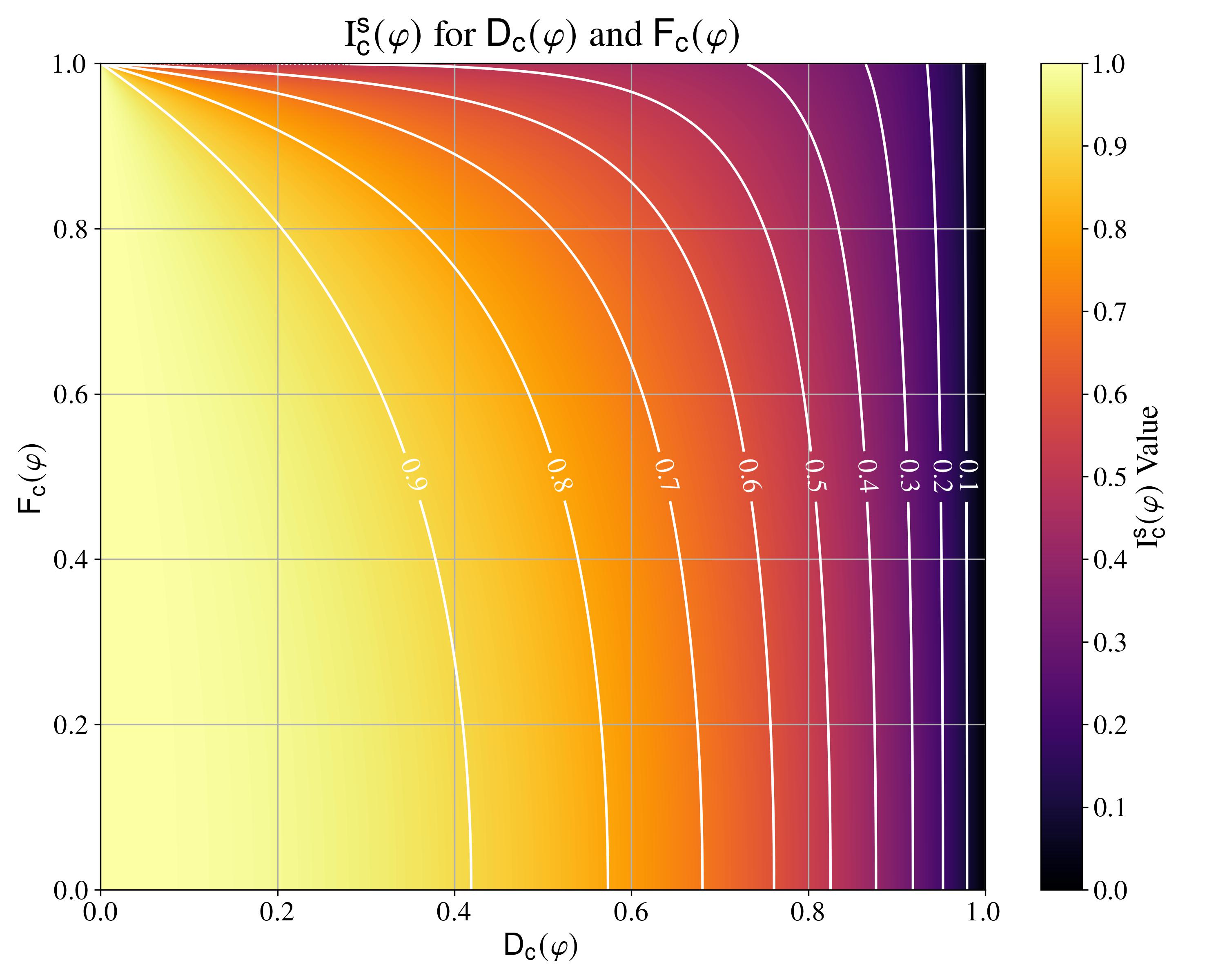}
         \caption{$\cis{\cchn}$ against $\cad \cchn$ and $\cfd \cchn$, for bit channels. Plotted from analytical expressions (see Appendix \ref{app-bit-chans-details}).} \label{fig:bit-cis-an}
\end{subfigure} 
\begin{subfigure}{0.45\textwidth}
\centering
\includegraphics[width=0.94\linewidth]{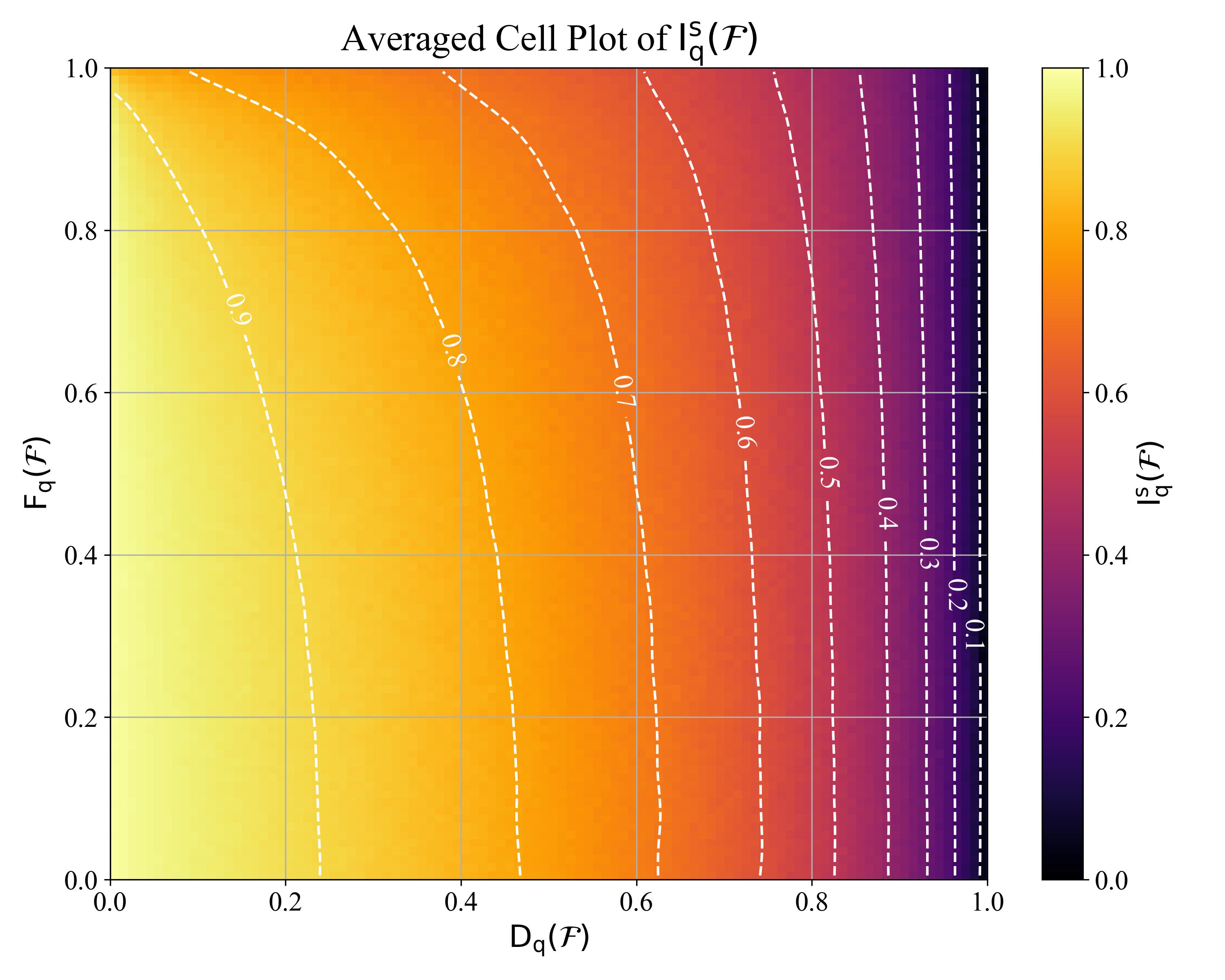}
\caption{Average $\qis{\chn}$ against $\qad \chn$ and $\qfd \chn$, for qubit channels.} \label{fig:qbit-qis}
\end{subfigure} 
\begin{subfigure}{0.45\textwidth}
\centering
         \includegraphics[width=0.94\linewidth]{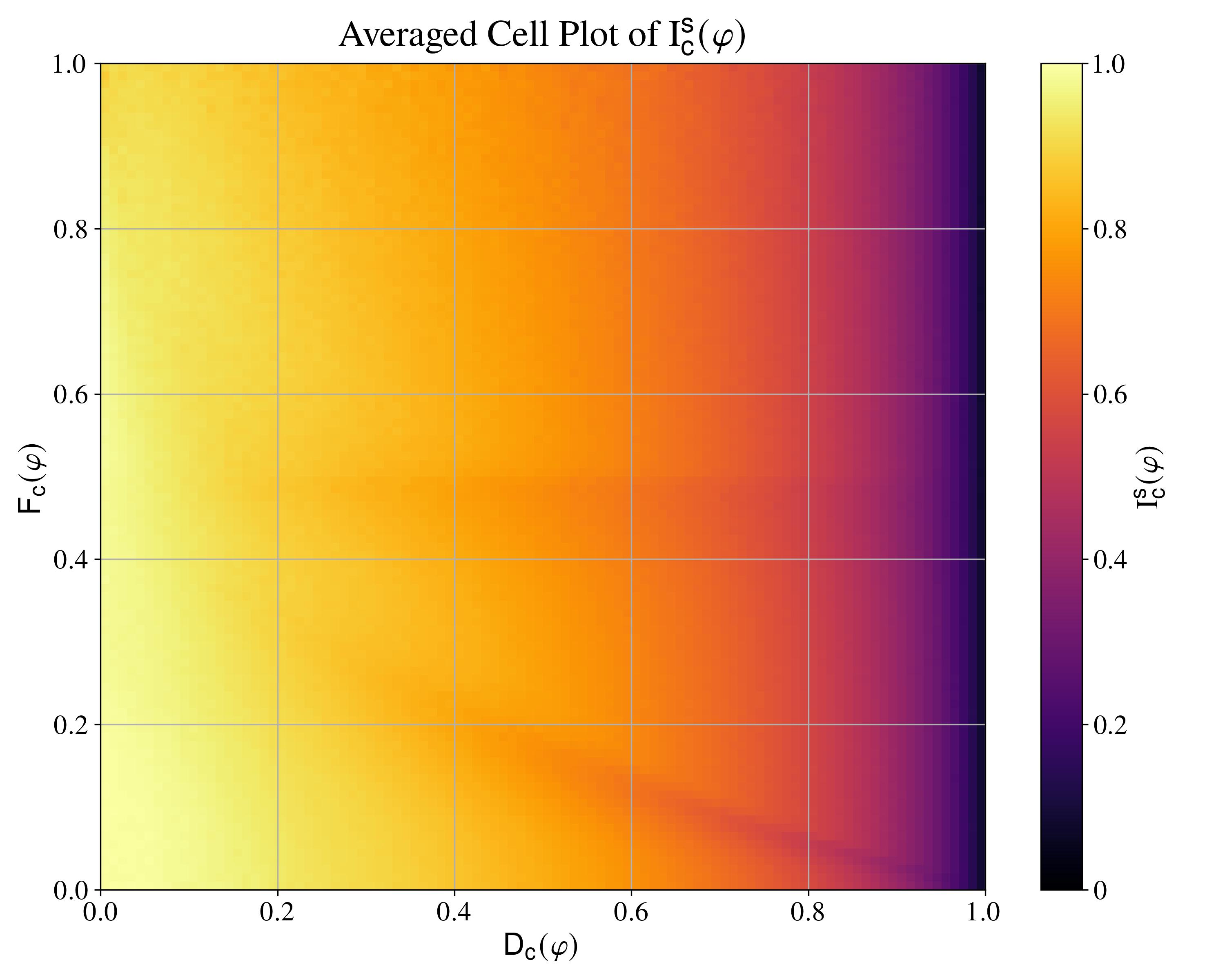}
 \caption{$\cis{\cchn}$ as a colour plot against $\cad \cchn$ and $\cfd \cchn$, for trit channels.} \label{fig:trit-cis}
\end{subfigure} 
\caption{Colour density plots for bit, qubit and trit channels of Bayesian subjectivity $\mathrm{I}^\mathsf{s}$ against absolute determinant and the fixed centroid distance.} \label{fig:subj-plots}
\end{figure}

\begin{figure}[h!] 
\begin{subfigure}{0.45\textwidth}
\centering
         \includegraphics[width=0.94\linewidth]{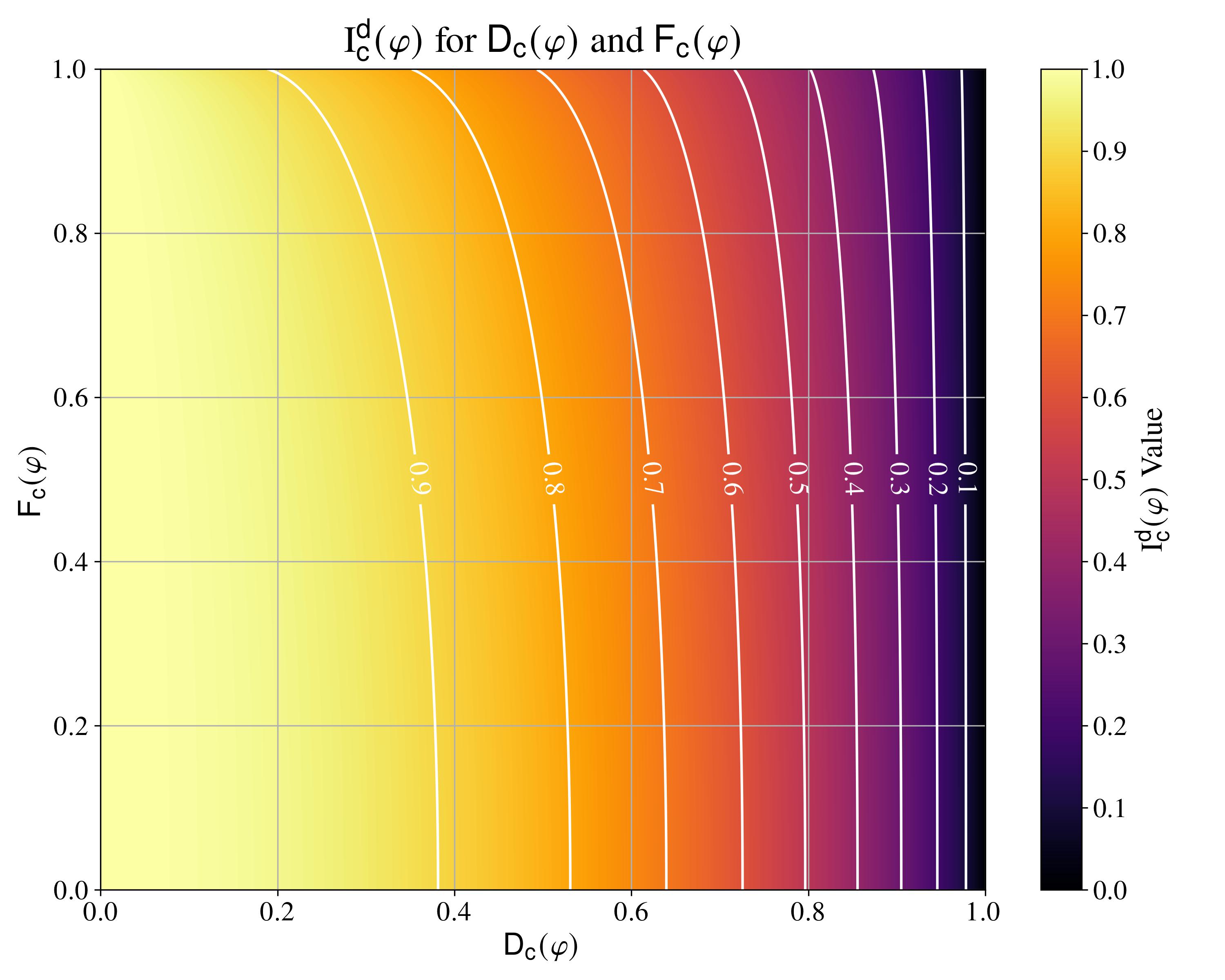}
         \caption{$\ccd{\cchn}$ against $\cad \cchn$ and $\cfd \cchn$, for bit channels. Plotted from analytical expressions (see Appendix \ref{app-bit-chans-details}).} \label{fig:bit-ccd-an}
\end{subfigure} 
\begin{subfigure}{0.45\textwidth}
\centering
\includegraphics[width=0.94\linewidth]{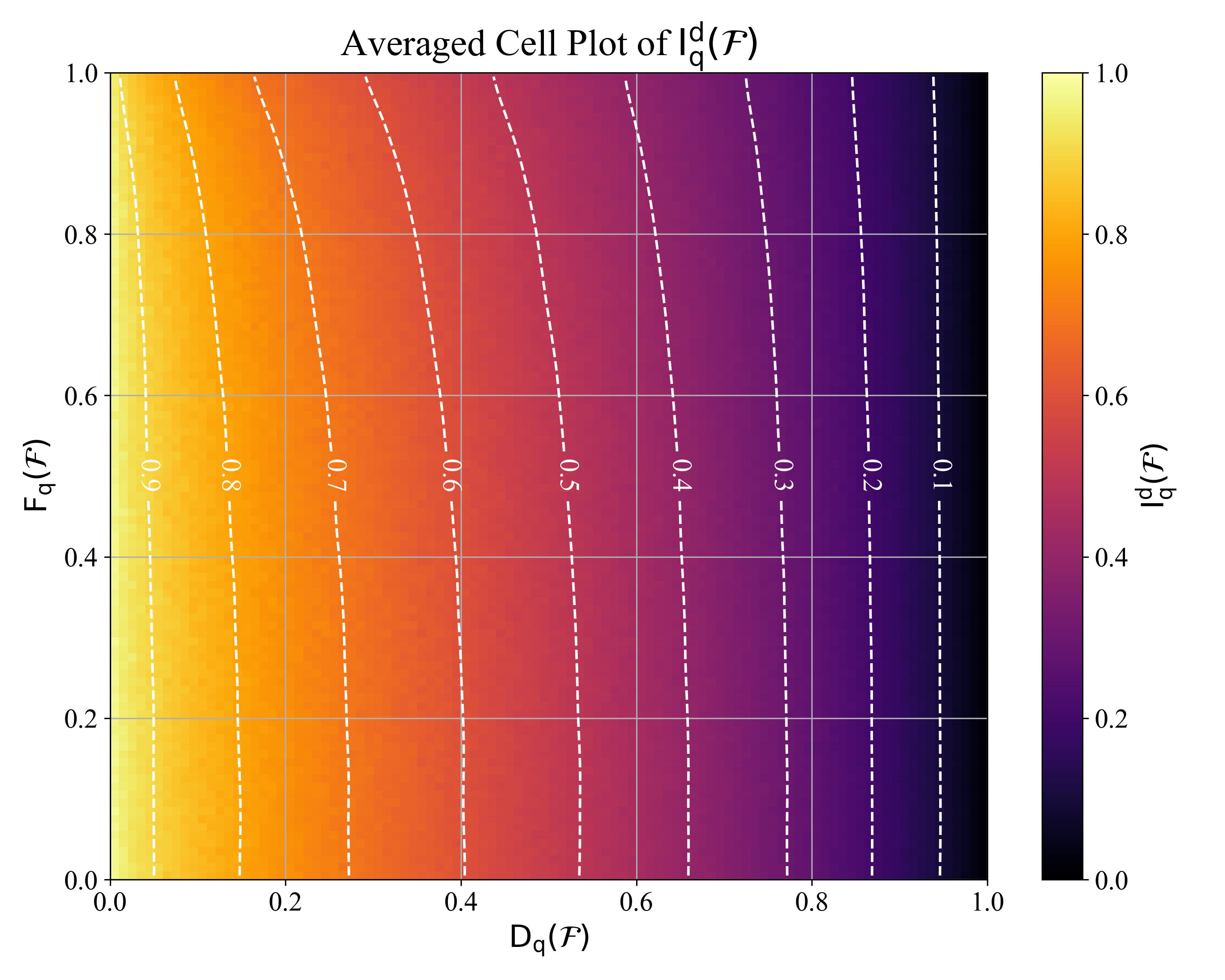}
\caption{Average $\qcd{\chn}$ against $\qad \chn$ and $\qfd \chn$, for qubit channels.} \label{fig:qbit-qcd}
\end{subfigure} 
\begin{subfigure}{0.45\textwidth}
\centering
         \includegraphics[width=0.94\linewidth]{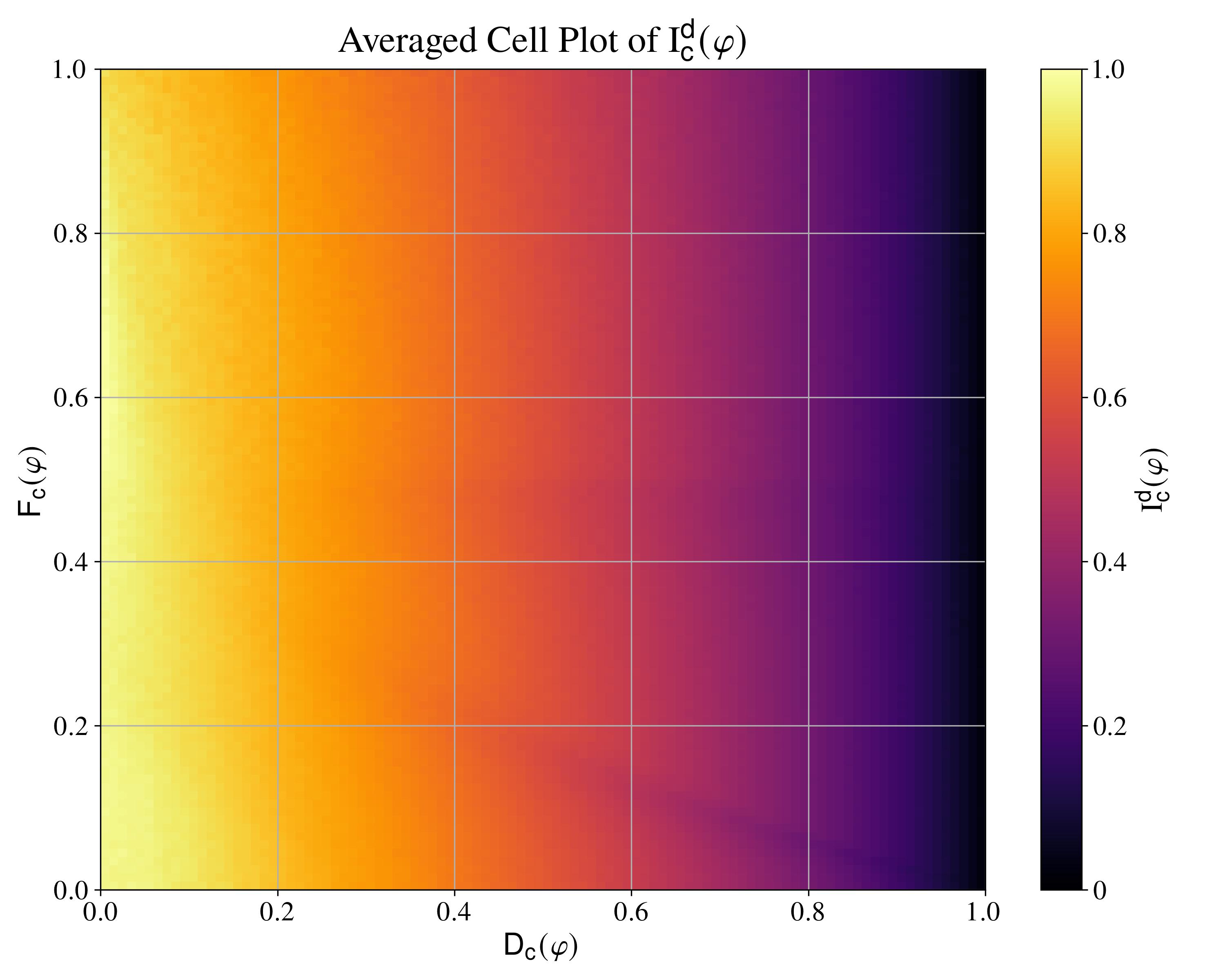}
 \caption{$\ccd{\cchn}$ as a colour plot against $\cad \cchn$ and $\cfd \cchn$, for trit channels.} \label{fig:trit-ccd}
\end{subfigure} 
\caption{Colour density plots for bit, qubit and trit channels of average change in divergence $\mathrm{I}^\mathsf{d}$ against absolute determinant and the fixed centroid distance.} \label{fig:change-in-div-plots}
\end{figure}

\subsection{Bit Channels}
From Theorems \ref{thm:adet-ref-rls} and \ref{thm:adet-er-ref-rls}, one may expect that there is a straightforward monotonic relationship between the absolute determinant of a channel and its subjectivity. It turns out, that even for simple bit channels, this does not obtain. Instead, we have found that there is a \textit{joint monotonicity} between $\cis\cchn$ and $\cad\cchn$ and a parameter we are calling the \textit{fixed centroid distance} $\cfd\cchn$. This is formally defined in \eqref{eq:cfd-def} but in this main text it suffices that we speak of its conceptual essence: $\cfd\cchn$ is simply the distance of the fixed point of $\cchn$ from the maximally mixed state (for details and definitions see Appendix \ref{fixed-centroid-sec}) and is a monotone of the \textit{purity} of the map's equilibrium.  

\subsubsection*{Joint monotonicity with \\ quasistaticity \& equilibrium entropy}

In the case of bit channels, the joint monotonicity of $\cis\cchn$ with respect to $\cad\cchn$ and $\cfd\cchn$ is seen plainly in the analytical plot in Figure \ref{fig:bit-cis-an}. One observes from the figure that if $\cad\cchn$ is fixed, $\cis\cchn$ increases as $\cfd \cchn$ decreases. Likewise, if one fixes $\cfd{\cchn}$, $\cis\cchn$ 
increases monotonically with increases in $\cad \cchn$. 

Now, since one may see $\cad\cchn$ a mathematical quantifier of quasistaticity (see Appendix \ref{relaxation-time-ssec}) and $\cfd\cchn$ is a quantifier of the entropy at steady state (see Appendix \ref{steady-state-equi-ssec}), this means our irreversibility measure $\cis\cchn$ properly captures two physical intuitions from thermodynamics: (1) the more approximately quasistatic a process is, the less irreversible it is and (2) the less entropy producing a process is, the less irreversible it is as well.  

We also remark that Theorems \ref{thm:bij-abs-det-irrv-measure} and \ref{thm:erasure-irr-measure-all-equal} are vividly demonstrated in Figure \ref{fig:bit-cis-an}. Furthermore, all points of the extreme right correspond to every erasure channel in this space of maps. From this, we see our conjecture that erasures always maximize $\cis\cchn$ is also met.

\subsubsection*{Comparison to Average Change in Divergence}
We also compare subjectivity to another measure which automatically satisfies the data processing inequality. We introduce this ``control'' or ``foil'' irreversibility measure here for both classical and quantum maps, calling it \textit{average change in divergence} $\mathrm{I}^\mathsf d$. It is defined as follows:
\begin{eqnarray}
    \ccd\cchn &=& \int\mathsf{div_c}(p||\rf) - \mathsf{div_c}(\cchn[p]||\cchn[\rf]) \, dp \, d \rf \label{eq:ccd-def} \\
    \qcd\chn &=& \int \mathsf{div_q}(\rho||\rf) - \mathsf{div_q}(\chn[\rho]||\chn[\rf])  \, d\rho \, d\rf,  \label{eq:qcd-def}
\end{eqnarray}
where $\mathsf{div_c}(p||\rf)$ is the Kullback-Leibler divergence between distributions $p$ and $\rf$ and $\mathsf{div_q}(\rho||\rf) = \Tr[\rho \log \rho -\rho \log \rf]$ is the Umegaki quantum relative entropy between density operators $\rho$ and $\rf$. The integrand describes irreversible entropy production when the evolution is a linear map (see e.g.~\cite{kolwol21}). In the classical case at least, $\gamma$ is indeed the Bayesian prior used to define the reverse process, if the entropy is defined as $\ln(P_F/P_R)$; the formal extension to the quantum case is discussed in Ref.~\cite{fullyquantum}. So the average change in divergence can also be called average entropy production.

One can verify that $\ccd\cchn$ and $\qcd\chn$ satisfy equivalent statements of Theorems \ref{thm:bij-abs-det-irrv-measure}, \ref{thm:erasure-irr-measure-all-equal} and \ref{thm:concat-preserve} (see Appendix \ref{app-avg-div-equiv-thms}). Hence, we normalize it the way we did for subjectivity \eqref{eq:normalization-irrev-measures}. Similarly, $\ccd\cchn$ displays joint monotonicity to $\cad\cchn$ and $\cfd{\cchn}$ for bit channels. As with Figure \ref{fig:bit-cis-an}, we plot a colour density graph, Figure \ref{fig:bit-ccd-an}, for $\ccd\cchn$. Comparing these two figures, one can see how $\cis\cchn$ is much more sensitive than $\ccd\cchn$ in the high $\cfd\cchn$ regime. This regime corresponds to a class of channels of special interest: \textit{absorbing maps} (see Appendix \ref{sec:absorbers}). We return to this point (and other comparisons between $\mathrm{I}^\mathsf{d}$ and $\mathrm{I}^\mathsf{s}$) in later discussions. For now, it suffices to say that $\cis\cchn$ is not equivalent to $\ccd\cchn$ to render it trivial, while also having enough similarities to strengthen its viability as an irreversibility measure.
\begin{center}
    \begin{figure}[h!] 
\begin{subfigure}{0.46\textwidth}
\centering
         \includegraphics[width=0.95\linewidth]{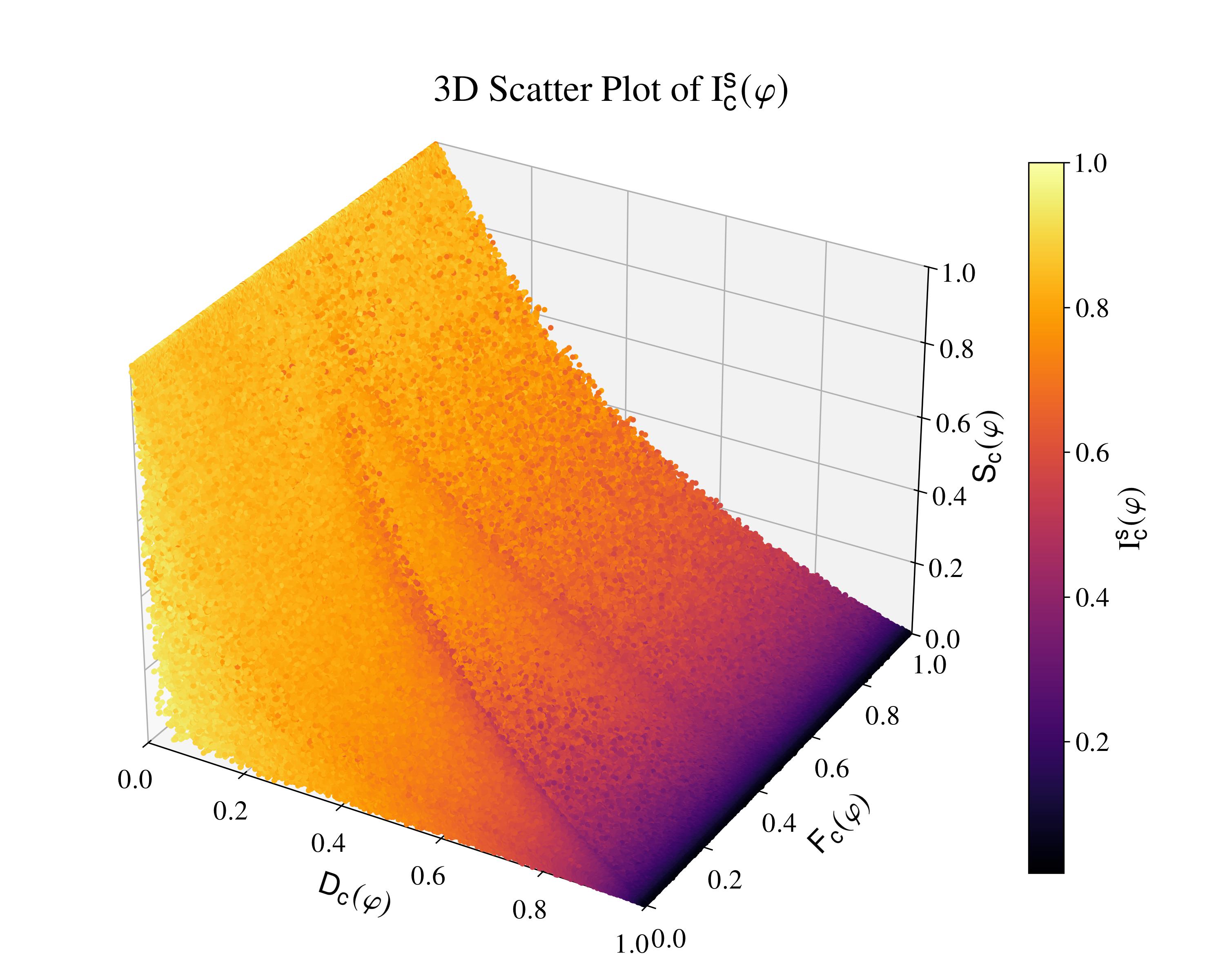}
         \caption{$\cis{\cchn}$ as a 3D colour density plot against $\cad \cchn$, $\cfd \cchn$ and $\skw(\cchn)$, for trit channels.} \label{fig:trit-cis-3d}
\end{subfigure} 
\begin{subfigure}{0.46\textwidth}
\centering
         \includegraphics[width=0.95\linewidth]{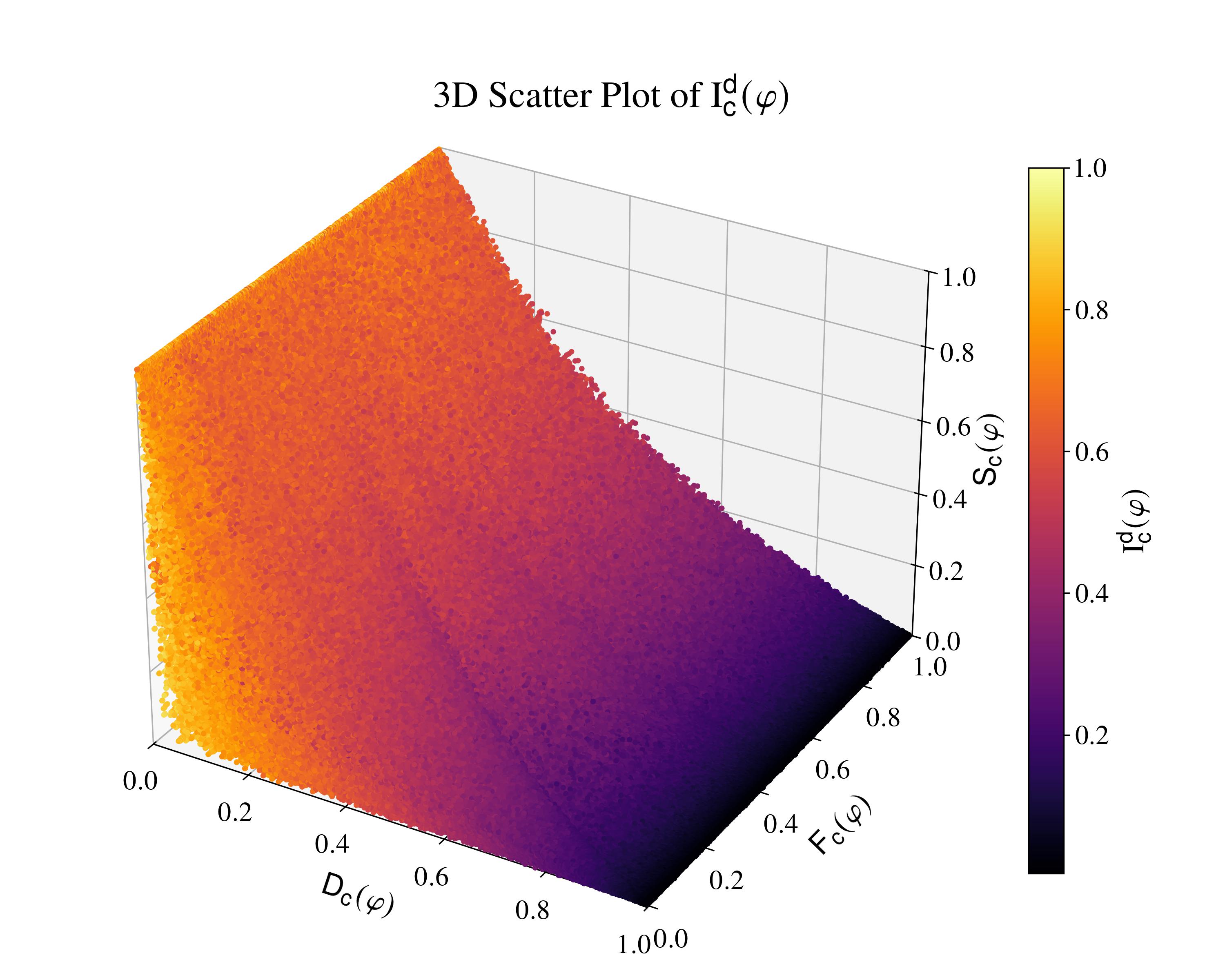}
         \caption{$\ccd{\cchn}$ as a 3D colour density plot against $\cad \cchn$, $\cfd \cchn$ and $\skw(\cchn)$, for trit channels.} \label{fig:trit-ccd-3d}
\end{subfigure} 
\begin{subfigure}{0.45\textwidth}
    \centering
    \includegraphics[width=0.93\linewidth]{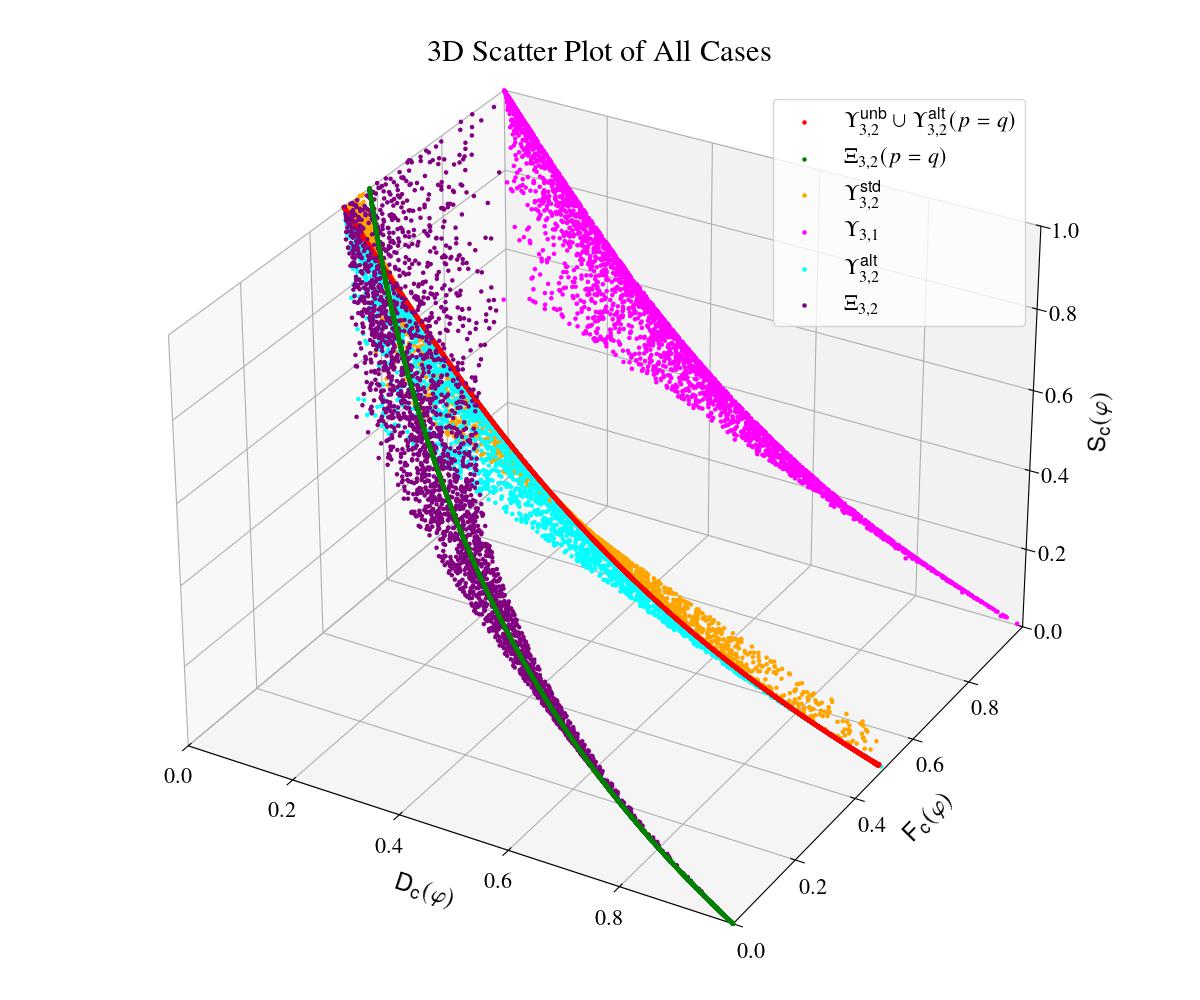}
    \caption{Notable trit channels are isolated to explain features of figures above. The red  and green lines correspond to the ``bisecting'' and ``arcing ridge'' respectively.}
    \label{fig:ridges}
\end{subfigure} 
\caption{3D colour plots for trit channels, from numerical methods. See Section \ref{ssec:trit} and Appendix \ref{app-explain-the-ridges}.} 
\label{fig:trit-result-3D}
\end{figure} 
\end{center}
\subsection{Qubit Channels}
We turn now to \textit{qubit} channels. We plot the average Bayesian subjectivity $\langle \qis\chn\rangle$ against $\qad{\chn}$ and $\qfd{\chn}$ in Figure \ref{fig:qbit-qis}. We do the same for $\langle\qcd\chn\rangle$ in Figure \ref{fig:qbit-qcd}. Taking the mean is necessary since $\qad{\chn}$ and $\qfd{\chn}$ underdetermines qubit channels, which take 12 real parameters to define \cite{choi1975choi-theorem,homa2024choi}. Due to this very large parameter space, we no longer see a strong, joint monotonic relationship between the parameters of the plots. That said, we do see significant \textit{heuristic} evidence that this joint monotonic relationship (in comparison to Figure \ref{fig:bit-cis-an}), on average or in some typical sense, carries over to the quantum regime for $d=2$. The numerical results also agree with Theorems \ref{thm:bij-abs-det-irrv-measure} and \ref{thm:erasure-irr-measure-all-equal}, as well as our conjecture that all $\cis\chn$ is maximum if and only if $\chn$ is a quantum erasure. As with the case of bit channels, Figures \ref{fig:qbit-qis} and \ref{fig:qbit-qcd} show how the two irreversibility measures differ. Once again, $\qis\chn$ appears to be more sensitive in the higher $\qfd\chn$ regime. $\qcd\chn$ seems more uniform for each choice of $\qad\chn$, while $\qis\chn$ is less so. Future work can be done to expand these observations and explore what other quantum-theoretic parameters are at play in determining values of $\qis\chn$. 

\subsection{Trit Channels} \label{ssec:trit}
In Figure \ref{fig:trit-cis}, we plot $\langle\cis{\cchn}\rangle$ against $\cad{\cchn}$ and $\cfd{\cchn}$, for trit channels $\cchn \in \mathbb{R}^3$, randomly sampling over the $6$ variables that define trit channels (details in Appendix \ref{app-random}). Once again Theorems \ref{thm:bij-abs-det-irrv-measure} and \ref{thm:erasure-irr-measure-all-equal} are instantiated, but while there is certainly some resemblance to Figure \ref{fig:bit-cis-an}, there are interesting differences. In particular, we notice two streaks of low subjectivity. 

To have a better picture of this, we introduce a new geometric measure in order to see a 3D diagram and explain these peculiarities. Since the simplexes for the classical trits are generally triangles, we opt for the \textit{skew} $\skw(\cchn)$ or irregularity of the channel's image. It shall be defined in the following way:
\begin{eqnarray}
    \skw(\cchn)= \max&\bigg(&\frac{\max(\theta^\cchn_{12},\theta^\cchn_{23},\theta^\cchn_{31})}{120\degree}, \nonumber \\ &&\frac{60\degree -\min(\theta^\cchn_{12},\theta^\cchn_{23},\theta^\cchn_{31})}{60\degree}\bigg)-\frac{1}{2},
\end{eqnarray}
where $\theta^\cchn_{ij} = \arccos\left[\frac{(\mt{\cchn}_i-\mt{\cchn}_k)\cdot (\mt{\cchn}_j-\mt{\cchn}_k)}{|\mt{\cchn}_i-\mt{\cchn}_k|_\mathsf{2} |\mt{\cchn}_j-\mt{\cchn}_k|_\mathsf{2}} \right]$. One can see how $\skw(\cchn) \in [0,1)$ with the lower bound corresponding to equilateral triangles of every size and the upper bound corresponding to any isosceles triangle with a vertex angle approaching $\pi$ radians. 

Setting $\skw(\cchn)$ as a new degree of freedom, we expand Figure \ref{fig:trit-cis} to \ref{fig:trit-cis-3d}. We see that the two streaks from the 2D colour plot goes to two ``ridges'' in the 3D plot. We remark here that this also occurs (though less distinctly) in our plot of $\ccd\cchn$ in Figure \ref{fig:trit-ccd-3d}. 

What explains these regions of unexpectedly high reversibility? We will not encumber the main text by getting into the technicalities. Suffice to say that, as with the behavior at high $\cfd\cchn$ in Figure \ref{fig:bit-cis-an}, these ridges correspond to \textit{absorbing channels} $\cabs_{d,n}$, and what might be called \textit{pseudo-absorbers} $\Xi_{d,n}$. This is identified visually in Figure \ref{fig:ridges}. We discuss these in detail in Appendix \ref{sec:absorbers}. The upshot is that the ridges occur due to the deterministic transitions that exist in these maps. Indeed, for these maps there exists $a,a'$ and $b\neq a,b'\neq a'$ such that $\cchn(a'|a)=\varphi(b'|b)=1$. The structure of these two unique deterministic transitions results in an embedded reversibility detected by $\cis\cchn$ and $\ccd\cchn$ which, in turn, form the ridge features. 

All this to say that these irreversibility measures (though more starkly in the case of Bayesian subjectivity) do capture reversibility features beyond the geometric properties we plotted before ($\cad\cchn$, $\cfd\cchn$), and pick up these subtler forms of reversibility. 

\section{Conclusions}\label{sec:concl}
Having reviewed analytical and numerical results, we draw this study on the irreversibility of classical and quantum maps to a close. Theorems \ref{thm:adet-ref-rls} and \ref{thm:adet-er-ref-rls} set the stage for this work's central research question: can we quantify the irreversibility of a linear transformation in terms of its Bayesian inversion's dependence on priors? We answered this question with the measure of Bayesian subjectivity $\mathrm{I}^\mathsf{s}$ given by \eqref{eq:cis-def} and \eqref{eq:qis-def}. 

Analytical results in Theorems \ref{thm:bij-abs-det-irrv-measure}, \ref{thm:erasure-irr-measure-all-equal} and \ref{thm:concat-preserve} show that the measure adheres, at every dimensionality, to physical intuitions for the extremes of reversibility and irreversibility. We have also discussed numerical evidence that Bayesian subjectivity (1) obeys the data processing equality, (2) detects physically intuitive properties like quasistaticity and entropy production (doing so in both the classical and quantum regimes) and (3) is sensitive to subtler sources of reversibility, such as those that resulted in the ridge-like features in the case of trits, seen vividly in Figure \ref{fig:trit-cis-3d}.  

We have also made comparisons between our subjectivity measure $\mathrm{I}^\mathsf{s}$ and the average change in divergence $\mathrm{I}^\mathsf{d}$, defined in \eqref{eq:ccd-def} and \eqref{eq:qcd-def}---which obeys the data processing inequality by construction. Our results show that subjectivity and this comparator measure are not trivially equivalent, but also feature sufficient similarities (such as the same properties for extreme maps for both the classical and quantum regime, the joint monotonicties for the bit channel case, and the detection of the ridges for the trits) so as to bolster the viability of $\mathrm{I}^\mathsf{s}$.

Together, we have seen strong evidence that the subjectivity of a map that is inherent to its retrodiction is a viable irreversibility, for both classically and quantum-theoretically modeled transformations. We emphasize that while Bayesian retrodiction is in general subjective (retrodiction depends on the choice of prior), the \textit{subjectivity} of a map is not itself subjective. As we have defined it, due to the integration over the space of all possible priors, the measure is finally prior-independent: Bayesian subjectivity is a function of the linear map alone. The measure \textit{objectively quantifies the subjectivity} inherent to Bayesian inference on each transformation. 

With this, we conclude by making some remarks on avenues for future work. The most obvious extension would be to analytically prove the data processing inequality for Bayesian subjectivity, as in \eqref{eq:c-dpi} and \eqref{eq:q-dpi}. If this can be done, it would imply that \textit{decreases in Bayesian subjectivity} over concatenations would constitute a non-Markovianity witness \cite{li2018-hierarchy-non-Markov, abiuso2023characterizing-non-Markov}. This has a nice interpretation: \textit{objective memory} from ``backflow'' decreases Bayesian \textit{subjectivity}. Ongoing numerical simulations have been promising in this regard, but the keystone of an analytic proof for the data processing inequality remains difficult to capture. Besides concerns about the inequality, one can also consider investigating what physically insightful quantum-theoretic parameters constitute joint monotonicity to quantum Bayesian subjectivity in the qubit regime and beyond. 

\vspace{1em}

{\em Acknowledgments} --- We thank Valerio Scarani, Lin Htoo Zaw, Ge Bai, Ray Garnadi, Nelly Ng, Paolo Abiuso, Hyukjoon Kwon and Haruki Emori for helpful discussions. This project is supported by the National Research Foundation, Singapore through the National Quantum Office, hosted in A*STAR, under its Centre for Quantum Technologies Funding Initiative (S24Q2d0009); and by the Ministry of Education, Singapore, under the Tier 2 grant ``Bayesian approach to irreversibility'' (Grant No.~MOE-T2EP50123-0002).

\clearpage 
\bibliography{refsret}

\begin{widetext}
\appendix

\section{Two Geometric Variables of Classical \& Quantum Linear Maps} \label{sec:geom}

As they become pertinent in our notes on how Bayesian subjectivity, as an irreversibility measure, satisfies intuitions about quasistaticity, information dissipation and entropy production, we explore two relevant geometric parameters that characterize every linear map. The first is the absolute determinant, that captures the preservation of the input state space, denoted by $\cad\cchn$ and $\qad\cchn$. The second is the fixed centroid displacement denoted by $\cfd\cchn$ and $\qfd\chn$. 

\subsection{Preservation of the State Space}

As illustrated in Figure \ref{fig:deform}, one can see how maps transform the state space that they act on. The relation between irreversibility and state space preservation is well studied according to several natively appropriate quantifiers of changes in state space volume \cite{hoover1994irr-to-vol,daems1999entropy-irr-to-vol,ramshaw2017entropy-irr-to-vol,gzyl2020geometry-entropy,wolfer2021information-geom-reverse-markov,van2021geometrical-info-irreversibility-markov,ito2018stochastic-thermo-info-geom,nicholson2018nonequilibrium-info-geom,kim2021information-geom-fluctuations-non-eq}. In the channel-theoretic framework we employ here, the most straightforward measure is the target map's determinant magnitude \cite{tao2012topics-magnitude-determinant,rodriguez2022optimization-magnitude-determinant-unitary} or the \textit{absolute determinant} of some matrix representing its the contractive action on the input state space. 

For classical maps, the absolute determinant $\cad\cchn$ is easily obtained since the map can be described by its stochastic matrix. For this, see \eqref{eq:cad} in the main text. This quantity gives the \textit{ratio} of the output space to the input space. For quantum maps, we first note that any channel $\chn$ acting on $d$-dimensional quantum systems corresponds, through affine transformation, to what is sometimes referred to as the channel's \textit{transfer matrix} \cite{nielsen2021gateTFRM,caro2024learningTFRM,hantzko2024pauliTFRM,roncallo2023pauliTFRM,greenbaum2015introductionTFRM}. The transfer matrix of such $\chn$ has entries:
\begin{equation}\label{eq:gtm-def}
    (\gtm{\chn})_{ij} = \frac{1}{d}\tr{\gpaul_i \chn[\gpaul_j^\dagger]}, \quad i,j\in\{1,2 \dots d^2\},
\end{equation}
where, for $k<d^2$, the $\gpaul_k$ operators are Hilbert-Schmidt orthogonal $\Tr[\gpaul_i \gpaul_j^\dagger] \propto \delta_{ij}$ and form some operator basis that can be used to define the generalized Bloch vector of any quantum state in $\mathbb{C}^d$. For $k = d^2$, we have the corresponding identity operator. As an affine representation, we get the expected $(\gtm{\chn})_{d^2d^2}=1$, while all other $(\gtm{\chn})_{d^2j}=0$. Likewise, translation of the state space's centroid is given by the translative component of $\gtm\chn$, which is the vector $\qfp{\chn}$ (we drop $\gpaul$ from the notation to avoid encumbrance) with entries:  
\begin{equation}\label{eq:q-translation-pt}
   \left(\qfp{\chn}\right)_{i} = (\gtm{\chn})_{i,d^2}= \tr{\gpaul_i \chn[\one/d]},
\end{equation} for $i \in \{1,2 \dots d^2-1\}$.
As an aside, when $d=2^n$, a canonical and orthonormalized choice would be the generalized Pauli basis: \begin{equation}
    \gpaul_k=\bigotimes_{s =1}^n \sigma_{k_s}, \quad \forall s : \; k_s \in\{1,2,3,4\},
\end{equation} with $\sigma_1,\sigma_2,\sigma_3$ corresponding to the Pauli matrices, and $\sigma_4 =\one$. Hence, for qubits, this reduces to 
\begin{equation}
    (\ptm{\chn})_{ij} = \frac{1}{2}\tr{\sigma_i \chn[\sigma_j]}, \quad i,j\in\{1,2 \dots 4\},
\end{equation}
and $\qfp{\chn}$ becomes the qubit channel's output given a maximally mixed state. 

With this, we have the preservation of the state space for quantum channels:
\begin{equation}\label{qad}
    \qad{\chn} = \frac{|\det(\gtm \chn)|}{|\det{\gtm{\one}}|}.
\end{equation}
The denominator is simply a normalization which goes to $1$ for orthonormal choices of operator basis $\{\gpaul_k\}$.  As with $\cad{\cchn}$ is, the state space is totally preserved if and only if $\qad{\chn}=1$ and, sends to $\qad{\chn}=0$ if and only if $\chn$ sends all states in $\mathbb{C}^d$ to some lower operator basis.

\subsubsection{Relation to Relaxation Time \& Quasistaticity} \label{relaxation-time-ssec}
The extent to which the state space is preserved by some map can be related to the quench or \textit{relaxation time} of its corresponding physical process. This is the time it takes for a process to send input states to equilibrium. This may be described in terms of how many iterations some map needs to be applied before all input states become constrained into a small state volume (or subspace) containing fixed points. With this, we may state the relaxation time $\rlx \in \mathbb{Z}^+$ of some map $\cchn$, defined for some precision $z \in \mathbb{R}^+$, as:
\begin{equation}\label{eq:tqdef}
    \rlx: \; \underset{t}{\arg\!\min} \; \cad{\cchn^t} \; \leq \; 10^{-z}
\end{equation}
This object's relationship with $\cad{\cchn}$ may be determined in the following way. Noting that for $x,y \in \mathbb{R} : \, |x^y|=|x|^y$ and that $\det(A^y) = \det(A)^y$, \eqref{eq:tqdef} goes to
\begin{eqnarray}
    \rlx: \; \;
    & \underset{t}{\arg\!\min} & \; \; \cad{\cchn}^t \; \leq 10^{-z} \\
    & \underset{t}{\arg\!\min} & \; \; t\log\cad{\cchn} \; \leq -z \\ 
    & \underset{t}{\min} & \, t \, \geq \, -\frac{z}{\log \cad{\cchn}}
\end{eqnarray}
We are dealing with stochastic maps, so $0\leq\cad{\cchn}\leq1$. Since $t$ is a positive, non-zero integer, this gives us: 
\begin{eqnarray}
    \rlx \, = \, \left\lceil - \frac{z}{\log \cad{\cchn}}\right\rceil_{>0}
\end{eqnarray}
$\rlx$ is thus a positive monotone of $\cad{\cchn}$. As $\cad{\cchn}$ decreases from $1$ toward $0$, $\rlx$ monotonically decreases from infinity (which corresponds to genuine quasistaticity) to its lower bound in one (which corresponds to instantaneous quenching). $\cad{\cchn} \mapsto \qad{\chn}$ gives the same principle, but for quantum channels.  

In this way, these objects quantifying information-theoretic spatial preservation can be understood as an expression of how long it takes for a physical process to bring the input states to equilibrium. In thermodynamics, this captures deviations from quasistaticity, resulting in irreversible entropy production. 

\subsection{Fixed Centroid Displacement} \label{fixed-centroid-sec}
While $\cad{\cchn}$ and $\qad{\chn}$ capture how rapidly maps send physical systems to equilibrium, we now consider geometric ways of describing the equilibrium itself. It is well-known that any stochastic map $\cchn$ (and any quantum channel $\chn$) has at least one fixed point $p_\tau$ ($\rho_\tau$) \footnote{These properties are established by the so-called Krylov-Bogoliubov argument \cite{kryloff1937theorie}, which involves the Bolzano-Weierstrass property \cite{oman2017short-Bolzano}. Furthermore, by the Perron-Frobenius theorem, a \textit{unique} steady state exists for irreducible matrices \cite{perron-frob-borobia1992geometric,shur2016detailed-perron,cheng2012note-perron}.}:
\begin{equation}
    \forall\cchn\;\exists p_\tau : \cchn[p_\tau] = p_\tau, \quad \forall\chn \; \exists \rho_\tau : \chn[\rho_\tau] = \rho_\tau.
\end{equation}
The bounds on the absolute determinant of these maps show that, unless $\cad\cchn$ (or $\qad{\chn}$) is unity, the iterations of the map eventually collapses all input states into a subspace of fixed points, which we call the map's \textit{fixed space}. This may not consist of a single point, but is a bounded space (since the state space is). Rather than tracking the whole fixed space, we track its centroid, which we call the \textit{fixed centroid}, denoted ($ \cfp{\ite \cchn}$, $\qfp{\ite\chn}$), where
\begin{equation}\label{eq:ite}
    \ite \cchn = \lim_{t\to\infty} \cchn^t, \quad \ite \chn = \lim_{t\to\infty} \chn^t, 
\end{equation} 
The fixed centroid is the fixed point to which the iteration of the map converges when the input is the maximally mixed state ($v^{\one/d}, \one/d$). Obviously, if a map's fixed point is unique, the fixed centroid will be that fixed point.

For the classical case, the fixed centroid for $\cchn$ is given by taking the mean of component columns $\{\mt{\ite \cchn}_i\}$ of the corresponding matrix $\mt{\ite \cchn}$:
\begin{equation}
    \cfp{\ite \cchn} = \frac{1}{d}\sum^d_{i=1} \mt{\ite \cchn}_i=\sum^d_{i=1}\mt{\ite \cchn}_iv^{\one/d}_i=S^{\ite \cchn} v^{\one/d}.
\end{equation}
This approach generally selects a unique state that can be used to obtain our desired measure, even when there are degeneracies in terms of eigenstates of $\mt{\cchn}$. However, when maps are non-stabilizing (i.e. $\cchn$ such that $\forall t\in \mathbb{Z}^+:\cchn^t \not\simeq
 \cchn^{t+1}$), more computationally costly, general approaches can be employed (see Appendix \ref{app:tech-cfd} for details). This is largely irrelevant for our studies of classical bits and trits. 

For the quantum case, the fixed centroid for $\chn$ can be inferred from \eqref{eq:q-translation-pt}, and expressed as such:
\begin{equation}
\left(\qfp{\ite\chn}\right)_{i} = \lim_{t \to \infty}(\gtm{\chn})^t_{i,d^2} = \tr{\gpaul_i \ite\chn[\one/d]}.
\end{equation}

For reversible maps, while $\ite \cchn$ and $\ite \chn$ are ill-defined, the maximally mixed state can be called their fixed centroid, since the iteration of such maps leaves it unchanged. Together, we now define the \textit{displacement of the fixed centroid} as follows:
\begin{eqnarray}
    \label{eq:cfd-def} \cfd{\cchn} = \frac{|\cfp{\ite \cchn}-\vr{\one/d}|_2}{|\vr\Ydown-\vr{\one/d}|_2}, \quad
    \qfd{\chn} = |\qfp{\ite \chn}|_2, 
\end{eqnarray}
where $\vr{\Ydown}$ is any vector that is pure (i.e. with a single $1$-entry and all zeros) and $|\vr{x}-\vr{y}|_2$ is the Euclidean $\ell_2$-norm between two vectors $\vr{x},\vr{y}$. This is such that $\cfd{\cchn}, \qfd{\chn} \in [0,1]$ obtains for every given $d$-dimensional setting. 


\subsubsection{Relation to Steady State \& Purity} \label{steady-state-equi-ssec}
Technicalities about the multiplicity of fixed points aside, it is true for the overwhelming majority of physical maps that $\cfd{\cchn}$ and $\qfd{\chn}$ correspond to the distance between the state space's center and the map's equilibrium or, more generally, its steady state. Physically, this captures how \textit{purifying} a process is. When $\cfd{\cchn}$ (or $\qfd{\chn}$) is closer to unity, the more well-defined the steady state is. For a thermodynamic analogy, lower values of $\cfd{\cchn}$ (or $\qfd{\chn}$) correspond to states for which the target system has increased in entropy, due to interactions with a hot environment for instance.

\subsubsection{A Technical Note for Non-Stabilizing Maps}\label{app:tech-cfd}
While \eqref{eq:cfd-def} avoids the ambiguity that emerges from multiple fixed points, it fails to obtain a unique value for non-stabilizing maps (i.e. $\cchn \text{ s.t. } \forall \mathsf{t} \in \mathbb{Z}^+ : \cchn^\mathsf{t}  \not\simeq
 \cchn^{\mathsf{t}+1}$). This makes it such that $\mt{\cchn}\cfp{\ite \cchn}\neq \cfp{\ite \cchn}$. This occurs when $\cchn$ sends the state simplex to a subspace of rank larger than one, and then permute the pure states in that subspace. These are ``alternating absorbers'', given by \eqref{eq:cabs-gen} with $\Phi_n \neq \one_n$. In this sense, $\cfp{\ite \cchn}$ is not generally a fixed point. In order to identify a unique value of $\cfd{\cchn}$ for such channels, we could always algorithmically select for an average over $\cfp{\cchn^{t+s}}$ for a large $\mathsf t$ and $\mathsf s\in[0,r-1]$ where $r$ is the number required for $\cfp{\cchn^{t+r}}=\cfp{\cchn^{t}}$. This allows us to resolve both the issue of degeneracy in fixed points and the issue of non-stabilizing channels. 

For the purposes for our investigations however, \eqref{eq:cfd-def} is the same as such a definition for essentially all sampled maps. For bit channels, such maps do not exist. For trits, non-stabilizing maps occur only such that an average over $\cfp{\cchn^{t+s}}$ is the same as the unique eigenstates of those maps \cite{perron-frob-borobia1992geometric,cheng2012note-perron,shur2016detailed-perron}. This is because non-stabilizing maps only have multiple fixed points when the permutation cycles in the absorbing space are disconnected. For the case of trits, absorbing spaces are either rank-$1$ or rank-$2$, so any permutation across the pure absorbing states are always connected, thus yielding unique fixed points. For these particular cases, we will use this more general approach (that is, taking unique fixed point)  in order to disambiguate between choices of $\cfp{\ite \cchn}$.

\section{Proofs for Various theorems}
Here we include various proofs (or proof sketches) for certain claims in the main text.
\subsection{Proof for Lemma \ref{res:q-erase}}\label{app-qer}
It is straightforward to see that the adjoint \eqref{eq:adjdil} of the swap channel must be a swap channel.
With that, the Petz reduces in the following way: 
\begin{align*}
\forall  \bullet \, \hat{\qer}_\rf[\bullet] & = \sqrt{\rf} \;  \TrB \left[\sqrt{\one \otimes  \tau } U^\dag \left(\frac{1}{\sqrt{\qer[\rf]}} \bullet \frac{1}{\sqrt{\qer[\rf]}} \otimes \one\right) U \sqrt{\one \otimes  \tau }\right]\sqrt{\rf} \\
  & =\sqrt{\rf} \;  \TrB \left[\sqrt{\one \otimes \tau} U^\dag \left(\frac{1}{\sqrt\tau} \bullet \frac{1}{\sqrt\tau} \otimes \one\right) U \sqrt{\one \otimes \tau}\right]\sqrt{\rf}\\
  & =  \sqrt{\rf} \;  \TrB \left[\sqrt{\one \otimes \tau} \;  \one \otimes\frac{1}{\sqrt\tau} \bullet \frac{1}{\sqrt\tau} \; \sqrt{\one \otimes \tau}\right]\sqrt{\rf} = \, \rf
\end{align*}
Which is (\textbf{I}) $\to$ (\textbf{II}). Now, (\textbf{II}) $\to$ (\textbf{I}) is proven as follows:
\begin{align*}
    \forall(\rf,\bullet) \quad & \petz[\bullet] =  \,\rf  \\
    & \sqrt{\rf} \; \chn^\dagger\!\left[\frac{1}{\sqrt{\chn[\rf]}} \bullet \frac{1}{\sqrt{\chn[\rf]}}\right] \sqrt{\rf} =   \,\rf  \\
    & \chn^\dagger\!\left[\frac{1}{\sqrt{\chn[\rf]}} \bullet \frac{1}{\sqrt{\chn[\rf]}}\right] = \,  \one  \\
    \forall(\rf,\bullet,\rho) \quad & 
    \Tr\left[ \chn^\dagger\!\left[\frac{1}{\sqrt{\chn[\rf]}} \bullet \frac{1}{\sqrt{\chn[\rf]}} \right] \rho \right] = \Tr\left[ \frac{1}{\sqrt{\chn[\rf]}} \bullet \frac{1}{\sqrt{\chn[\rf]}} \chn[\rho] \right], \quad \because \eqref{eq:adj} \\    
    \forall(\rf,\rho) \quad & \Tr \Biggr[\underbrace{\chn^\dagger\!\left[\frac{1}{\sqrt{\chn[\rf]}} \rho \frac{1}{\sqrt{\chn[\rf]}} \right]}_{  \one} \rho \Biggr] = \Tr\left[  \rho \frac{1}{\sqrt{\chn[\rf]}} \chn[\rho] \frac{1}{\sqrt{\chn[\rf]}}\right] \\
    & 
    \Tr\left[\rho \left(   \one - \frac{1}{\sqrt{\chn[\rf]}} \chn[\rho] \frac{1}{\sqrt{\chn[\rf]}} \right) \right] = 0 \quad \text{noting $\forall \rho$} \\
    &  \one - \frac{1}{\sqrt{\chn[\rf]}} \chn[\rho] \frac{1}{\sqrt{\chn[\rf]}} = 0 \quad \Rightarrow \quad \forall(\rf,\rho) \; \; \chn[\rho] = \chn[\rf]  \\
    & \Rightarrow \forall \bullet \quad \chn[\bullet] = \tau
\end{align*}
Finally, (\textbf{II}) $\to$ (\textbf{III}) holds trivially, (\textbf{III}) $\to$ (\textbf{II}) invokes, as with the classical case, the recoverability property of the Petz map \cite{petz,petz1,wilde-recov,petzisking2022axioms}: $\forall \rf: \petz\circ \chn [\rf] = \rf$. With (\textbf{III}), since $\chn [\rf]$ is an instantiation of $\rho$, this recoverability feature implies that the fixed point $\mu$ \textit{must} be $\rf$. Which concludes our proof for Theorem \ref{res:q-erase}.

\subsection{Proofs for Theorem \ref{thm:concat-preserve}} \label{app-proof-concat-cis}
\begin{proof}
    For \eqref{rln1}, consider $\Omega = \Phi \circ \cchn$. By composability, $\hat\Omega_\rf= \rvp \circ \hat\Phi_{\cchn[\rf]} = \rvp \circ \Phi^\text{T}$. And so each disagreement element becomes $|| \hat\Omega_{\rf_1}-\hat\Omega_{\rf_2} ||_\lambda=|| (\hat\cchn_{\rf_1}-\hat\cchn_{\rf_2})\Phi^\text{T}||_\lambda=|| \hat\cchn_{\rf_1}-\hat\cchn_{\rf_2}||_\lambda$, since eigenvalues are not changed the action of a bijection. Thus, as per \eqref{eq:cis-def}, \eqref{rln1} holds. Similarly, for the quantum theoretic equivalent in \eqref{rln2}, the diamond norm between two channels is preserved if both channels are transformed by the same unitary action. Thus, \eqref{rln2} holds. As for \eqref{rln3} and \eqref{rln4}, we simply note that any general transformation acting after an erasure (that erases to some choice of fixed point $\tau$) is still an erasure: $\chn\circ \qer_\tau = \qer_{\chn[\tau]}$, $\cchn\circ \cer_\tau=\cer_{\cchn[\tau]}$. So via Theorem \ref{thm:erasure-irr-measure-all-equal}, this gives \eqref{rln3} and \eqref{rln4}. 
\end{proof}

\subsection{For Equivalent Statements of Theorems \ref{thm:bij-abs-det-irrv-measure}, \ref{thm:erasure-irr-measure-all-equal} and \ref{thm:concat-preserve} for Average Change in Divergence $\mathrm I^\mathsf{d}$} \label{app-avg-div-equiv-thms}

\begin{theorem} \label{ccdthm3} The average change in divergence is zero if and only if the map is bijective (or unitary),
\begin{eqnarray*}
    \ccd\cchn=0 \; &\Leftrightarrow& \; \cchn \text{ is a bjiection.} \\
   \qcd\chn=0 \; &\Leftrightarrow& \; \chn \text{ is a unitary.} 
\end{eqnarray*}
\end{theorem}

\begin{proof}
The ``if'' direction is straightforward.  For $\ccd{\mathcal{U}}$, one simply notes how unitaries preserve the divergence between any two density operators, which gives $\ccd{\mathcal{U}}=0$. 
As for the ``only if'' direction, the data processing inequality assures that the integrand of $\ccd \chn$ must always be non-negative This means $\ccd \chn =0$ implies $\forall (\rho,\sigma):  \mathsf{div_q}(\chn[\rho]||\chn[\sigma])=\mathsf{div_q}(\rho||\sigma)$, which is only fulfilled when $\chn$ is unitary, as they uniquely saturate the inequality.
\end{proof}

\begin{theorem} \label{ccdthm4}
For any given choice of dimension, erasures for the state space of that dimension always share the same values of average change of divergence:
\begin{eqnarray*}
    \forall(\cer_1,\cer_2)\in \mathbb{R}^d &:&  \ccd{\cer_1} = \ccd{\cer_2} \\
   \forall(\qer_1,\qer_2)\in \mathbb{C}^d &:&  \qcd{\qer_1} = \qcd{\qer_2} 
\end{eqnarray*}
\end{theorem}

\begin{proof} 
    $\forall \qer:\qcd\qer= \int \mathsf{div_q}(\rho||\sigma)  \, d\rho \, d\sigma$ as $\mathsf{div_q}(\qer[\rho]||\qer[\sigma])=0$ always as it always erases to the same state. Hence, $\qcd\qer$ is also independent of the choice of $\qer$. The proof for the classical case goes the same.
\end{proof}
 
\begin{theorem}\label{thm:concat-preserve-ccd}
Bijections (or unitaries) always preserve the average change in divergence \eqref{eq:ccd-def}, and the average change in divergence of an erasure channel can never be increased. 
    \begin{eqnarray}
    \forall(\Phi,\cchn) &:&  \ccd{\Phi\circ\cchn} = \ccd{\cchn} \\
    \forall(\mathcal{U},\chn) &:& \qcd{\mathcal{U}\circ\chn} = \qcd{\chn} \\
    \forall(\cer,\cchn) &:&  \ccd{\cchn\circ\cer} = \ccd{\cer} \\
    \forall(\qer,\chn) &:&  \qcd{\chn\circ\qer} =\qcd{\qer}
    \end{eqnarray}
\end{theorem}
\begin{proof}
    Bijections and unitaries always preserve the relative entropies between states (i.e $\mathsf{div_c}(\Phi[p]||\Phi[q])= \mathsf{div_c}(p||q)$, $\mathsf{div_q}(\mathcal U [\rho]||\mathcal U[\eta])= \mathsf{div_c}(\rho||\eta)$). So the first two relations are shown easily. The other two relations hold for the exact same reasons as their counterparts in Theorem \ref{thm:concat-preserve} (see Appendix \ref{app-proof-concat-cis}).
\end{proof}

\section{Analytical expressions of $\cis\cchn$ and $\ccd\cchn$ against $\cad\cchn$ and $\cfd\cchn$, for $\cchn \in \mathbb{R^2}$}\label{app-bit-chans-details}

\begin{eqnarray}
\mathrm{I}^\mathsf{s}_\mathsf{c}(\varphi) &=& 
f\left( \frac{\mathsf{D_c}(\varphi)}{1 + \sqrt{2}\,\mathsf{F_c}(\varphi)\,(1 - \mathsf{D_c}(\varphi))} \right)
+ f\left( \frac{\mathsf{D_c}(\varphi)}{1 - \sqrt{2}\,\mathsf{F_c}(\varphi)\,(1 - \mathsf{D_c}(\varphi))} \right) \\[0.9em]
&\text{where }& f(z) = \left(z^{-2} - 1\right) \left[ 1 - \left(z^{-2} - 1\right) \operatorname{arctanh}^2(z) \right] \nonumber \\ \nonumber \\
\mathrm{I}^\mathsf{d}_\mathsf{c}(\varphi) &=& \frac{\mathsf{D_c}(\varphi)}{4} \left[
f\left( \frac{\mathsf{D_c}(\varphi)}{1 + \sqrt{2}\,\mathsf{F_c}(\varphi)\,(1 - \mathsf{D_c}(\varphi))} \right)
+ f\left( \frac{\mathsf{D_c}(\varphi)}{1 - \sqrt{2}\,\mathsf{F_c}(\varphi)\,(1 - \mathsf{D_c}(\varphi))} \right)
\right] \\[0.9em]
&\text{where }& f(z) = \left(z^{-2} - 1\right) \operatorname{arctanh}(z). \nonumber \\ \nonumber 
\end{eqnarray}

\section{Random sampling for qubits and trits in Section \ref{sec:resul-discuss}}
\label{app-random}

\subsection{Methods for Qubit Channels}




Fix a grid cell $[u_i,u_{i+1})\times[f_j,f_{j+1})\subset[0,1]^2$ and draw
$u\sim{\rm Unif}[u_i,u_{i+1})$, $f\sim{\rm Unif}[f_j,f_{j+1})$; map to GAD parameters by
$\gamma=1-\sqrt{u}$ and $p=\tfrac{1\pm f}{2}$ (the sign chosen with probability $1/2$). 
Let the ancillary state be $\beta=\mathrm{diag}(p,1-p)$ and the two–qubit dilation
\[
U_{\rm AD}(\gamma)=
\begin{pmatrix}
1&0&0&0\\
0&\sqrt{1-\gamma}&\sqrt{\gamma}&0\\
0&-\sqrt{\gamma}&\sqrt{1-\gamma}&0\\
0&0&0&1
\end{pmatrix},
\qquad 
\mathcal E_0(\rho)=\operatorname{Tr}_B\!\big[U_{\rm AD}(\gamma)\,(\rho\otimes\beta)\,U_{\rm AD}(\gamma)^\dagger\big].
\]
In affine form, $T_{\rm GAD}=\mathrm{diag}\big(\sqrt{1-\gamma},\sqrt{1-\gamma},1-\gamma\big)$ and 
$t_{\rm GAD}=(0,0,\gamma(2p-1))$, hence $|\det T_{\rm GAD}|=(1-\gamma)^2=u$ and 
$\|(I-T_{\rm GAD})^{-1}t_{\rm GAD}\|=|2p-1|=f$. 
Apply a unitary sandwich on the system side with a fixed probability to enrich within–cell variability,
\[
U=(U_1\!\otimes I)\,U_{\rm AD}(\gamma)\,(U_2\!\otimes I),\qquad U_1,U_2\in U(2),\qquad
\mathcal E(\rho)=\operatorname{Tr}_B\!\big[U\,(\rho\otimes\beta)\,U^\dagger\big],
\]
under which $T'=R(U_1)\,T_{\rm GAD}\,R(U_2)$ and $t'=R(U_1)\,t_{\rm GAD}$ with $R(\cdot)\in SO(3)$, so 
$|\det T'|=|\det T_{\rm GAD}|=u$ while $\|(I-T')^{-1}t'\|$ may vary. 
Accept the sandwiched sample iff $\big(|\det T'|,\ \|(I-T')^{-1}t'\|\big)\in [u_i,u_{i+1})\times[f_j,f_{j+1})$; otherwise resample $(U_1,U_2)$ (finite retries) or revert to the unsandwiched $U_{\rm AD}(\gamma)$. 
Persist the final pair $(U,\beta)$ along with the derived coordinates, and repeat until each grid cell reaches its quota.

\subsection{Methods for Trit Channels}

For a given lower bound of determinant value equal to \(D\), we define the restricted simplex as
\begin{equation}
\{ (r_1, r_2) \in [0, 1-D]^2 \mid r_1 + r_2 \leq 1-D \}.
\end{equation}
To generate a trit channel, we construct the random matrix
\begin{equation}
M^{\cchn} = \begin{pmatrix}
s & a & p \\
t & b & q \\
1-s-t & 1-a-b & 1-p-q
\end{pmatrix},
\end{equation}
where each pair \((s,t)\), \((a,b)\), and \((p,q)\) is independently sampled uniformly from the restricted simplex defined above.

\begin{center}
\begin{figure}[h!]
    \centering
\includegraphics[width=0.7\linewidth]{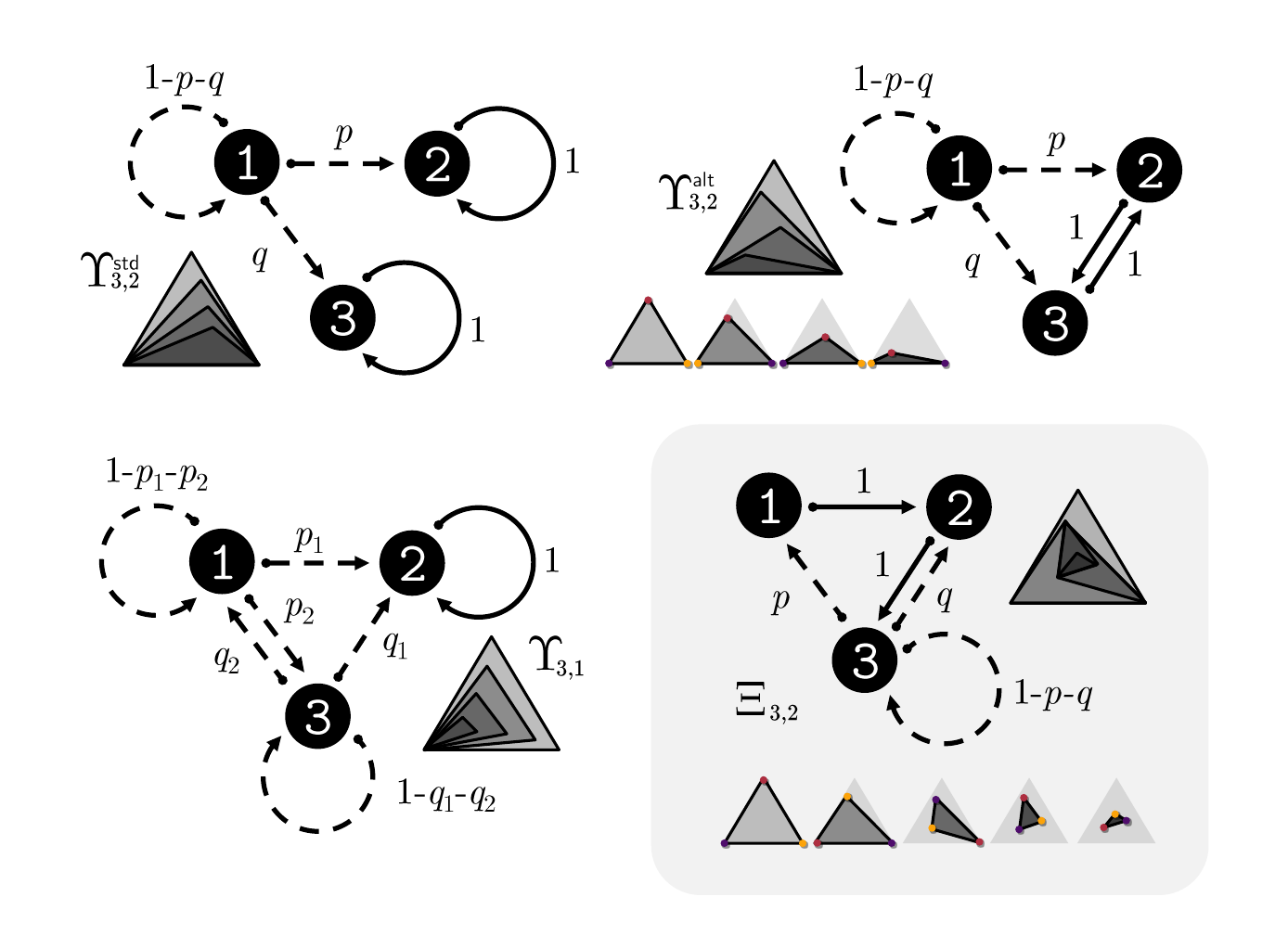}
    \caption{Illustrations (in terms of Markov models and channel images on the state simplex) for various kinds of absorbing maps $\cabs_{3,n}$ for trit processes. $\spi$ are pseudo-absorbing \eqref{eq:spiral-def} and becomes relevant when explaining the so-called ``bisecting ridge'' found in Figure \ref{fig:trit-cis-3d} and \ref{fig:trit-ccd-3d}.}
    \label{fig:special-maps}
\end{figure}
\end{center}

\section{Absorbing Channels \& Related Discussions}\label{sec:absorbers}
In the main text (particularly when discussing Figures \ref{fig:bit-cis-an} and \ref{fig:trit-cis-3d}, we alluded to the role of a special family of maps. These are a superset of erasure and partial erasure channels. These are \textit{absorbing} maps \cite{kemenysnell1969finite-absorbingmaps}. We discuss these in some detail here, introducing them generally and noting their characteristics under Bayesian retrodiction. In Appendix \ref{app-explain-the-ridges}, we focus into explaining the ridge-like features in Figure \ref{fig:trit-cis-3d}. 

\subsection{Absorbing Maps \& Properties Under Bayesian Inversion}
We first define what we mean by absorbing maps. 
\begin{itemize}
    \item Classical absorbing maps $\cabs_{\!d,n}$ are those that, over an arbitrarily large number of iterations, erase a $d-$dimensional space to a subspace that still includes some $n$ vertices (pure states) of the state space, but strictly less than the number of vertices of the original state space ($n<d$). Formally, for any $\cchn \in \mathbb{R}^d$,
    \begin{eqnarray}\label{eq:cabs-def} 
     \rank (\mt{\ite{\cchn}}) =  n<d \quad
     \wedge \quad  \left(\textstyle{\bigvee_{ij}^d} \; \mt{\cchn}_{ij}=1\right) \quad
     \Leftrightarrow   \quad \cchn \text{ is classically absorbing},
    \end{eqnarray}
    where $\vee$ refers to logical disjunction, and so $\textstyle{\bigvee_{ij}^d} \; \mt{\cchn}_{ij}=1$ is simply saying that there are some transitions in $\cchn$ that are deterministic. 
    \item Quantum absorbing maps $\qabs_{d,n}$ are those where states in the span of a basis in a $\mathbb{C}^d$ (i.e. some generalized Bloch hypersphere) become reduced, over an arbitrarily large number of iterations, to a subspace spanned by new basis with cardinality $n<d$:
    \begin{eqnarray}
        \ite \chn:\mathcal{B}(\mathbb{C}^d) \mapsto \mathcal{B}(\mathbb{C}^{n<d}) \quad \Leftrightarrow \quad \chn \text{ is quantum absorbing},
    \end{eqnarray}
    where $\mathcal{B}(\mathbb{C}^x)$ is the space of bounded linear operators on the Hilbert space $\mathbb{C}^x$. 
\end{itemize} One can consider Figure \ref{fig:special-maps} for illustrations of the geometric action of such maps that act on a trit space. Now, the $n$-dimension spaces are called the \textit{absorbing spaces}. The exclusion of the original state space and absorbing space is the \textit{transient space} of dimension $m=d-n$. In this way, absorbing maps go to at least partial erasure channels after many iterations.

\subsubsection*{On Classical Absorbing Maps}

Any map $\cabs_{\!d,n}^\mathsf{std}$ written in the following way is a classical absorbing channel for states in a $d$-dimensional space:
\begin{equation}\label{eq:cabs-usual}
\mt{\cabs_{\!d,n}^\mathsf{std}} = \left(
\begin{array}{r|cc}
\one_n & \tfrm   \\ \hline
\mathbb{0}_{m,n}\! & \trnm
\end{array}\right) 
\end{equation}
where, alongside the column stochasticity of $\mt{\cabs_{\!d,n}}$, the following must hold:
\begin{itemize}
    \item $\mathbb{0}_{m,n}$ is an $m \times n$ zeroes matrix 
    \item $\tfrm$ is an $n \times m$ matrix for which there is at least one non-zero entry,
    \item $\trnm$ is an $m\times m$ matrix for which $\det(\one_{m} -\trnm)\neq 0$. This is all to say that the $\trnm$ is any viable block for which its difference with an identity matrix of the same size is not singular. Of course, in such a way that the overall matrix remains stochastic.
\end{itemize}

Now \eqref{eq:cabs-usual} can be seen as absorbing channels that ``damp'' toward to lower indexed states, we clarify in passing that \eqref{eq:cabs-def} is generally fulfilled by a broader definition of absorbing maps:
\begin{equation}\label{eq:cabs-gen}
\mt{\cabs_{\!d,n}} = \mt{\Phi_d}\left(
\begin{array}{r|cc}
\mt{\Phi_n} & \tfrm   \\ \hline
\mathbb{0}_{m,n}\! & \trnm
\end{array}\right) (\mt{\Phi_d})^\tpose
\end{equation}
which frees how we define the states that compose the absorbing space through some choice of $d$-dimensional permutation ${\Phi_d}$, while allowing for deterministic transitions within the absorbing space through some other choice of permutation ${\Phi_n}$. Due to our geometric emphasis and since these essentially boil down to relabeling of states, we will largely go with the expression used in \eqref{eq:cabs-usual} for our proofs, though this more general form will be mentioned when relevant. We note a subclass of classical absorbing maps that may be referred to as \textit{deterministic absorbers} $\cabs_{\!d,1}$:
\begin{equation} \label{eq:det-abs}\mt{\cabs_{\!d,1}} = \left(
\begin{array}{r|cc}
1  & \tfrr \\ \hline
 \vr{\mathbb{0}}  & \trnm 
\end{array}\right)
\end{equation}
where, $\vr{\mathbb{0}}$ is a zero vector (with $d-1$ entries in this case) and $\tfrr$ is row with at least one non-zero entry. These maps go at least asymptotically close to an erasure channel with a maximum fixed centroid distance (for that dimensionality of channels), \begin{equation} \label{eq:dabs-cfd}
    \cfd{{\cabs}_{\!d,1}}=\max_{\cchn \in \mathbb{R}^{d}}\, \cfd{\cchn} = 1
\end{equation}

This geometric note can be extended to general classical absorbers $\cabs_{\!d,n}$:
\begin{theorem}\label{thm:min-cfd-abs}
The fixed centroid displacement of a classical (d,n)-absorbing map is bounded as such,
\begin{equation}\label{eq:min-cfd-abs}
\sqrt{\frac{d-n}{(d-1) n}}   \leq \cfd{{\cabs}_{\!d,n}} < \sqrt{\frac{(d-n)(d+1-n)}{(d-1)d}}
\end{equation}
\end{theorem}
\begin{proof}
    Absorbing maps reduce to some lower subspace. The fixed centroid distance will be minimized at the midpoint of such a subspace, which corresponds to a uniform distribution in that $n$-dimensional subspace. Thus, take \eqref{eq:cfd-def}, with $\mathsf{p_c}^{\!\ite{\cabs}_{\!d,n}}$ given by a uniform distribution on the absorbing space only. It simplifies to lower bound above. The upper bound is when the whole transient space goes asymptotically close to absorbing to \textit{one} of the vertices. This sends the space to a $\mathsf{p_c}^{\!\ite{\cabs}_{\!d,n}}$ as close as possible to that vertex, which is also furthest away possible to the centroid of the state simplex. This is given by a vector that has a single $(d-n+1)/d$ entry, $n-1$ entries with $1/d$, and zeroes for everything else. This gives a $\cfd{{\cabs}_{\!d,n}}$ that simplifies to the upper bound above. 
\end{proof}
Note how the result agrees with \eqref{eq:dabs-cfd}.

\subsubsection*{On Quantum Absorbing Maps}
Quantum absorbing maps $\qabs_{d,n}$ are maps that send all states from $\mathbb{C}_d$ to $\mathbb{C}_{1\leq n<d}$. For our investigation, we simply take special note of quantum deterministic absorbers for qubits $\qabs_{2,1}$, which is defined for any qubit unitary $U_2$ and transition weight $s$ in this way:
\begin{eqnarray}
    & \qabs_{2,1}[\rho] =\sum_{x=0}^1 K_x \rho K_x^\dagger , \quad
    & K_0 = U_2 
    \begin{pmatrix}
    1 & 0 \\
    0 & \sqrt{1-s} 
    \end{pmatrix}, \quad 
    K_1 = U_2 
    \begin{pmatrix}
    0 & \sqrt{s} \\
    0 & 0
    \end{pmatrix}.
\end{eqnarray}
For $U_2=\one_2$, $\qabs_{2,1}$ models a damping or spontaneous emission, sending all states to $\ketbra{0}{0}$. In any case, we have for generalized deterministic quantum absorbing maps,
\begin{equation}
    \qfd{\bar{\qabs}_{d,1}}=\max_{\chn \in \mathbb{C}^{d}}\, \qfd{\chn}=1,
\end{equation}
as they send every state to some pure state, which is as far as possible from the maximally mixed state, for that space of states.
\subsubsection*{Bayesian Inference on Absorbing Channels}
Here, we briefly consider the Bayesian inversions of absorbing channels. We first consider the classical case. The matrix representation of $(\hat\cabs_{\!d,n})_\rf$, by \eqref{bayes}, is as follows:
\begin{equation}\label{eq:cabs-ret}
    \mt{(\hat\cabs_{\!d,n})_\rf} = 
    \left(\begin{array}{r|c}
    D^{\rf/\cabs_{\!d,n}[\rf]}_n &  \mathbb{0}_{n,m}  \\ \hline
    \tfrm^{\rf} \! & \trnm^{\rf}
    \end{array}\right),
\end{equation}
where, the following holds
\begin{itemize}
    \item $D^{\rf/\cchn[\rf]}_n$ is a diagonal matrix with $ii$-entries $\rf(i)/\cchn[\rf](i)$ for $i\in\{1,2\dots n\}$, the absorbing space.
    \item $\tfrm^{\rf}$ is an $m \times n$ matrix that \textit{only} depends on $\rf$ and $\tfrm$. It represents the retrodictive transitions from the absorbing space of $\cabs_{\!d,n}$ into its respective transient space.
    \item $\trnm^{\rf}$ is an $m\times m$ matrix that \textit{only} depends on $\rf$ and $\trnm$. It represents the retrodictive transitions, internally speaking, of the transient space of $\cabs_{\!d,n}$.
\end{itemize}
We now make some notes on this Bayesian inverse of absorbing maps $(\hat\cabs_{\!d,n})_\rf$.  
\begin{theorem}
    $\cabs_{\!d,n}$ preserves zeroes in the transient space, 
    \begin{equation}
        \mt{\cabs_{\!d,n}}\left(\! 
        \begin{array}{cc}
        \vr{p}_{\in n} \\ \hline  \vr{0}_{\in m}  
        \end{array}\!\right) =
        \left(\! 
        \begin{array}{cc}
        \vr{p'}_{\in n} \\ \hline  \vr{0}_{m\in m}  
        \end{array}\!\right),   
    \end{equation}
    while $(\hat\cabs_{\!d,n})_\rf$ preserves zeroes in the absorbing space:
    \begin{equation}
         \mt{(\hat\cabs_{\!d,n})_\rf}\left(\! 
        \begin{array}{cc}
        \vr{0}_{\in n} \\ \hline  \vr{q}_{\in m}  
        \end{array}\!\right) =
        \left(\! 
        \begin{array}{cc}
        \vr{0}_{\in n} \\ \hline  \vr{q'}_{\in m}  
        \end{array}\!\right)
    \end{equation}
    where, $\vr{p}_{\in x}$ is a $x$-dimensional column vector representing a probability distribution $p$. 
\end{theorem}
\begin{proof}
    This simply follows from \eqref{eq:cabs-usual} and \eqref{eq:cabs-ret}. It also applies generally for more complicated structures of transient and absorbing spaces as per \eqref{eq:cabs-gen}.
\end{proof}
\begin{theorem} \label{thm:abs-to-abs}
    $\trnm$ of an absorbing channel (from $d \to n$) $\cabs_{\! d,n}$ is diagonal if and only if, for all reference priors, $(\hat\cabs_{\!d,n})_\rf$ is an absorbing channel (but from $d \to m$):
    \begin{eqnarray*}
        \cabs_{\!d,n}  \; \; \text{\normalfont{s.t.}} \; \;& \forall&(i,f): (\trnm)_{if} = \delta_{if} \, z(i) \\ \Leftrightarrow \quad  &\forall& \rf :(\hat\cabs_{\!d,n})_\rf \text{ is a $(d,m)$-absorbing map.}
    \end{eqnarray*}
\end{theorem}

\begin{proof}
    Firstly, we simply note that for any given $d,n$ that defines $\cabs_{\!d,n}$, it is implicit that $z(i) \neq 1$. This is because if any $z(i)=1$, then $d,n$ must change. Secondly $\tfrm^{\rf}$ being the identity is equivalent to $(\hat\cabs_{\!d,n})_\rf$ being an absorbing channel, except to the transient space as opposed to the original absorbing space of $\cabs_{\!d,n}$. $\trnm^{\rf}$ fulfills the role of a new transfer block through Bayes rule $\eqref{bayes}$, together with the fact that the original $\trnm$ always has a non-zero entry.

   Now, if $\trnm$ is \textit{diagonal}, then it follows that for the entries of $\trnm^{\rf}$, i.e. $i,f \in [n+1,d]$ entries of $(\hat\cabs_{\!d,n})_\rf$ (written here as $\cabs$), we have
    \begin{eqnarray}
        \hat{\cabs}_{\!\rf}(i|f) = \delta_{if}z(i) \frac{\rf(i)}{\cabs [\rf] (f)}
    \end{eqnarray}
    which implies,
    \begin{eqnarray*}
        \forall(i&\neq& f) \; \hat{\cabs}_{\!\rf}(i|f) =  0 \\
        \forall(i&=& f) \; \hat{\cabs}_{\!\rf}(f|f) =  \frac{\rf(f)z(f)}{\sum_x \cabs(f|x) \rf(x)}
        = \frac{\rf(f)z(f)}{\displaystyle{\underbrace{\sum_{x'=1}^n \cabs(f|x') \rf(x')}_{0}} +\displaystyle{\sum_{x=n+1}^d \underbrace{\cabs(f|x)}_{\delta_{fx}\,z(f)} \rf(x)}} = 1,
    \end{eqnarray*}
    which simply means that $\trnm^{\rf}$ is the identity for $m$-dimensions, which means $\hat{\cabs}_{\!\rf}$ is absorbing. 

    Conversely, if it is the case that for all $i,f \in [n+1,d]$, $\hat{\cabs}_{\!\rf}(i|f) = \delta_{if}$, then
    \begin{eqnarray}
       \forall \rf: \cabs(f|i)\frac{\rf(i)}{\cabs [\rf] (f)}= \delta_{if}.
    \end{eqnarray}
    This means that when $i\neq f$, $\cabs(f|i)=0$ so as to fulfill the condition for all $\rf$. Meanwhile, when $i = f$, we note again that $\sum_{x'=1}^n \cabs(f|x') \rf(x')=0$ for $f \in [n+1,d]$. Together, this gives, 
        \begin{eqnarray}
       \forall \rf: \frac{\cabs(f|f)\rf(f)}{\sum_{x=n+1}^d \cabs(f|x)\rf{(x)}}= 1.
    \end{eqnarray}
    This implies that $ \cabs(f|x)=\delta_{fx} z(f)$ for all $x,f \in [n+1,d]$, implying the diagonality of $\trnm$.
\end{proof}

Two notable corollaries follow. Firstly,
\begin{coro} \label{cor-det-abs}
    Every $(\hat\cabs_{d,d-1})_\rf$ is a deterministic absorbing channel for all choices of $\rf$.
\end{coro}
\begin{proof}
    $\cabs_{d,d-1}$ implies that $m=1$ and so $\trnm$ has a single entry and is therefore diagonal. Via Theorem \ref{thm:abs-to-abs}, this implies $(\hat\cabs_{d,d-1})_\rf$ is deterministically absorbing. 
\end{proof}

Relatedly, one finds that $\cabs_{2,1}$ (sometimes called Z-channels) give the Bayesian inverse $(\hat\cabs_{2,1})_\rf$ for some prior $\vr{\rf} = (
p \;\;\ 1-p )^\tpose$ in the following way:
\begin{eqnarray} \label{eq;z-chan}
            & \mt{\cabs_{\!2,1}}& = \left(\!
            \begin{array}{cc}
            1  & s \\ 
             0  & 1-s 
            \end{array}\!\right) 
            \; \Leftrightarrow \; \mt{(\hat\cabs_{2,1})_\rf} =
            \left(
            \begin{array}{cc}
             \frac{p}{p+(1-p)s} & 0 \\
             \frac{(1-p)s}{p+(1-p)s} & 1 \\
            \end{array}
            \right) 
            \; \Rightarrow \; \mt{(\hat\cabs_{2,1})_\rf} \left(
            \begin{array}{cc}
             q  \\
             1-q  \\
            \end{array}
            \right)  = \left(
\begin{array}{c}
 \frac{p q}{p+(1-p)s} \\
 1-\frac{p q}{p+(1-p)s} \\
\end{array}
\right) \label{eq:damp-classical} 
\end{eqnarray}
It becomes clear that,
\begin{coro}
   $\cabs_{2,1}$ are the only maps for which the forward map and Bayesian inverse are both deterministic absorbers, for all choices of reference prior $\rf$.
\end{coro}
\begin{proof}
    From Corollary \ref{cor-det-abs} and Theorem \ref{thm:abs-to-abs}.
\end{proof}
\begin{coro}\label{thm:cabs-entries}
     There are  at least $(d-1)n$ entries in ${(\hat\cabs_{\!d,n})_\rf}$ that are independent of $\rf$. If $\trnm$ of $\cabs_{\!d,n}$ is diagonal, there are at least $n^2 + d^2 -(1+d)n$ entries in ${(\hat\cabs_{\!d,n})_\rf}$ that are independent of $\rf$.
\end{coro}
\begin{proof}
    These follow from the number of entries in $(\hat\cabs_{\!d,n})_\rf$ that are $0$ or $1$, given \eqref{eq:cabs-ret} and Theorem \ref{thm:abs-to-abs}.
\end{proof}

As for the quantum case, for our investigation, it suffices to say that the qubit case of the amplitude damping channel $\qabs_{2,1}$ for $U_2=\one_2$ has a Petz recovery such that the classical reduction:
\begin{eqnarray}
    (\hat\qabs_{2,1})_\rf[\rho] = \left(
    \begin{array}{cc}
    \frac{p q}{p+(1-p)s} & 0\\
    0& 1-\frac{p q}{p+(1-p)s} \\
    \end{array} \right), \quad
    \text{where } \rho= \left(
    \begin{array}{cc}
    q & 0\\
    0& 1-q \\
    \end{array} \right), \; 
    \rf= \left(
    \begin{array}{cc}
    p & 0\\
    0& 1-p \\
    \end{array} \right), \nonumber
\end{eqnarray}
is consistent with \eqref{eq:damp-classical}. When $\rho$ or $\rf$ do not commute with the computational basis for which the damping is defined, then the output generally has coherence terms, that depend on the eigensystem of $\rf$. By symmetry, these principles of commutativity carry over when other $U_2$ are chosen.

\subsubsection*{Physical Analogy for Absorbing Maps}
Absorbing maps can be seen as an information-theoretic expression of chemically catalyzed processes \cite{masel2001chemical-catalyst,steinfeld1999chemical-catalyst} (not to be confused with the kind of catalysis defined for quantum resource-theoretic thermodynamics). The absorbing space are all states for which the product has been totally yielded. The transient space is where there still exist reagents. $\trnm$ contains stochastic transitions, enabled by the catalyst, modeling how the reagents are gradually processed into the product. $\tfrm$ contain the transitions for which the catalytic reaction is complete. $\Phi_n$, as in \eqref{eq:cabs-gen}, are any deterministic transitions within the completed space. These correspond to so-called work steps---processes with no thermal component. 

Catalysts aid both forward and reversed arrows of a physical process by lowering the activation energy regardless of the directionality. This might be seen as a physical analogue to Corollary \ref{thm:cabs-entries}.  

\subsection{Explaining Features of Figure \ref{fig:trit-cis-3d} and \ref{fig:trit-ccd-3d}} \label{app-explain-the-ridges}
Figures \ref{fig:trit-cis-3d} and \ref{fig:trit-ccd-3d} show that whatever $\cis\cchn$ captures in terms of irreversibility, it is going beyond the geometric properties we have put forward. The most striking features here are the two strands of noticeably high reversibility. The strand that cuts as a line through these plots at $\cfd\cchn=0.5$ can be called the ``bisecting ridge''. Meanwhile the strand that cuts through the figure from $\cad\cchn=1,\cfd\cchn=0$ and asymptotically to $\cad\cchn=0,\cfd\cchn=1/\sqrt3\approx0.577$ in a curve will be referred to as the ``arcing ridge''. Given that our irreversibility measures have drawn our attention to these segments, it is natural to ask if some subtler form of irreversibility may be at play. In particular, it is interesting to ask whether these features emerge from some specific kind of map that may be tied to some physical or information geometric characteristic related to reversibility. 

\subsubsection*{The Bisecting Ridge \& Uniform Absorbers}

\newcommand{\aabs}{\cabs_{3,2}^\mathsf{alt}}
\newcommand{\uabs}{\cabs_{3,2}^\mathsf{unb}}
\newcommand{\unia}{\cabs_{3,2}^\mathsf{uni}}

The bisecting ridge, via Theorem \ref{thm:min-cfd-abs}, marks the threshold of where absorbing maps occur in the space of maps, specifically in terms of $\cfd\cchn$. More fundamentally, this is where ``uniform absorbers'' $\unia$ reside. These refer to the specific class of $\cabs_{\!d,n}$ that absorb toward a $n$-dimensioned uniform prior, saturating the lower bound on $\cfd{\cabs_{d,n}}$ in \eqref{eq:min-cfd-abs}. There are two kinds of uniform absorbers for the trit case (i.e. $\{\unia\} =\{\aabs\} \cup \{\uabs\}$). The first is the \textit{alternating absorber} $\aabs$. This is simply \eqref{eq:cabs-gen} for $\cabs_{3,2}$ with $\Phi_n$ as the bit flip. For some $p\in[0,1-q]$, $q\in[0,1-p]$ and $(p,q) \neq (0,0)$:
\begin{equation}
    \aabs=\mt{\Phi_3}\left( 
\begin{array}{ccc}
 0 & 1 & p\\
 1 & 0 & q \\
 0 & 0 & 1-p-q
\end{array}\right)(\mt{\Phi_3})^\tpose.
\end{equation}
$\cfd{\aabs} = 0.5$ because, while the maps are non-stabilizing, they have a unique steady state: the uniform prior for a rank-$2$ subspace (consider Appendix \ref{app:tech-cfd}). This is due to the bit flipping between the pure states of the absorbing space. The high entropy fixed point emerges from this kind of irreversibility that compounds over many iterations. 

The second kind of uniform absorber is the \textit{unbiased absorber} $\uabs$. For $p\in (0,\frac{1}{2}]$, it is given by
\begin{equation}
    \uabs=\mt{\Phi_3}\left( 
\begin{array}{ccc}
 1 & 0 & p\\
 0 & 1 & p \\
 0 & 0 & 1-2p
\end{array}\right)(\mt{\Phi_3})^\tpose.
\end{equation}
These trit maps absorb toward a subspace's uniform prior simply by virtue of having equally weighted transitions in $\tfrm$. Now, about the ridge in general, we may mention some features of note:
\begin{enumerate}
    \item We numerically verified  that \textit{all} $\aabs$ and $\uabs$ fall in this slice of $\cfd\cchn=0.5$. Figures \ref{fig:ridges} and \ref{fig:uni-abs-slice} may be consulted. Here, these maps generally take higher values of $\skw(\cchn)$ for any given value of $\cad\cchn$. This large population of absorbing channels, as per Corollary \ref{thm:cabs-entries}, contributes to the lower $\cis\cchn$ for that slice of $\cfd\cchn$.
    \item The upper surface of this ridge in Figures \ref{fig:trit-cis-3d} and \ref{fig:trit-ccd-3d} corresponds to all $\uabs$ and also $\aabs$ for $p=q$. This may be seen in Figures \ref{fig:ridges} and \ref{fig:uni-abs-slice}. These are the points with the highest value of $\skw(\cchn)$ for all $\cchn$ such that $\cfd{\cchn}=0.5$.
    \item On average, our numerics in Figure \ref{fig:uni-abs-slice} reflect that $\gcm(\unia)$ decreases for lower values of $\skw(\unia)$. Put differently, as $|p-q|$ increases, $\gcm(\aabs)$ decreases. This is sensible as the transitions into the absorbing space become more differentiated, reversal becomes less subjective. 
\end{enumerate}

\subsubsection*{The Arcing Ridge \& Pseudo-absorbing Maps} 
The arcing ridge has particularly low $\cis\cchn$ largely due to the contribution of what might be called \textit{spiral} maps or pseudo-absorbers $\spi$. These superficially resemble absorbing channels as per \eqref{eq:cabs-gen}, but are not so because there are transitions out of the apparent absorbing space. For this reason, no choice of permutations $\Phi_3, \Phi_2$ for \eqref{eq:cabs-gen} will produce $\spi$. Formally, spiral maps are given by,
\begin{equation}\label{eq:spiral-def}
    \spi=\mt{\Phi_3}\left( 
\begin{array}{ccc}
 0 & 0 & p\\
 1 & 0 & q \\
 0 & 1 & 1-p-q
\end{array}\right)(\mt{\Phi_3})^\tpose,
\end{equation}
for $p\in (0,1-q], q\in[0,1-p]$. Due to the $1$-entries of spiral maps, the images of these maps always attach themselves to some vertices of the state simplex. This is a kind of pseudo-absorbing space, as the probability current transits out of this space after some iterations. Over these iterations the map spirals toward its fixed point, which can be anywhere on the simplex. Once again, we run through some features of note:
\begin{enumerate}
    \item Our numerics (see Figure \ref{fig:ridges}) show that all $\spi$ are clustered around the ridge, beginning from the point at $(\cad{\spi},\cfd\spi)=(1,0)$ to various values on a surface at $\cad{\spi} = 0$ and $\cfd\spi>0.5$, increasing in $\skw(\spi)$ as $\cad{\spi}$ falls.
    \item As with $\cabs_{3,2}$, $\skw(\spi)$ is maximized with $p=q$. Due to the two vertexes being fixed for the pseudo-absorbing space, these are also the maps with the highest $\cad\cchn$ for any  of $\skw(\cchn)$ and $\cfd\cchn$. Hence, the upper surface of the arcing ridge in Figures \ref{fig:trit-cis-3d} and \ref{fig:trit-ccd-3d} are also these unbiased $\spi$.
    \item In a similar vein as absorbing channels, all Bayesian inversions $(\hat\Xi_{3,2})_\rf$ on spiral maps are independent of $\rf$ on 5 out of 9 entries. This accounts for the low values of $\gcm(\spi)$ that characterize the ridge feature.
\end{enumerate}

\begin{figure}[b!] 
\begin{subfigure}{0.45\textwidth}
\centering         \includegraphics[width=1\linewidth]{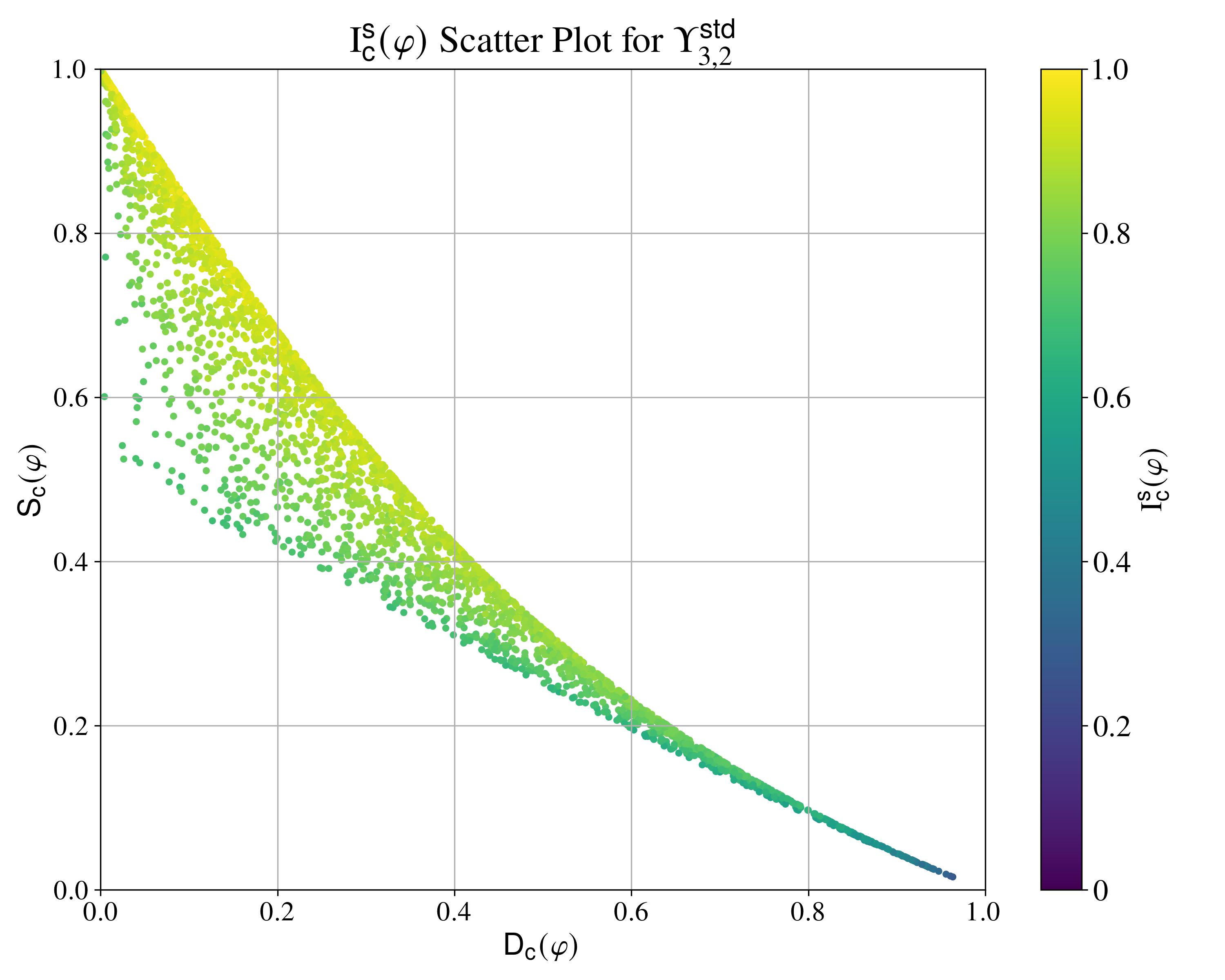}
         \caption{For $\cabs_{3,2}^{\mathsf{std}}$.} \label{fig:trit-stdabs}
\end{subfigure} 
\begin{subfigure}{0.45\textwidth}
\centering
         \includegraphics[width=1\linewidth]{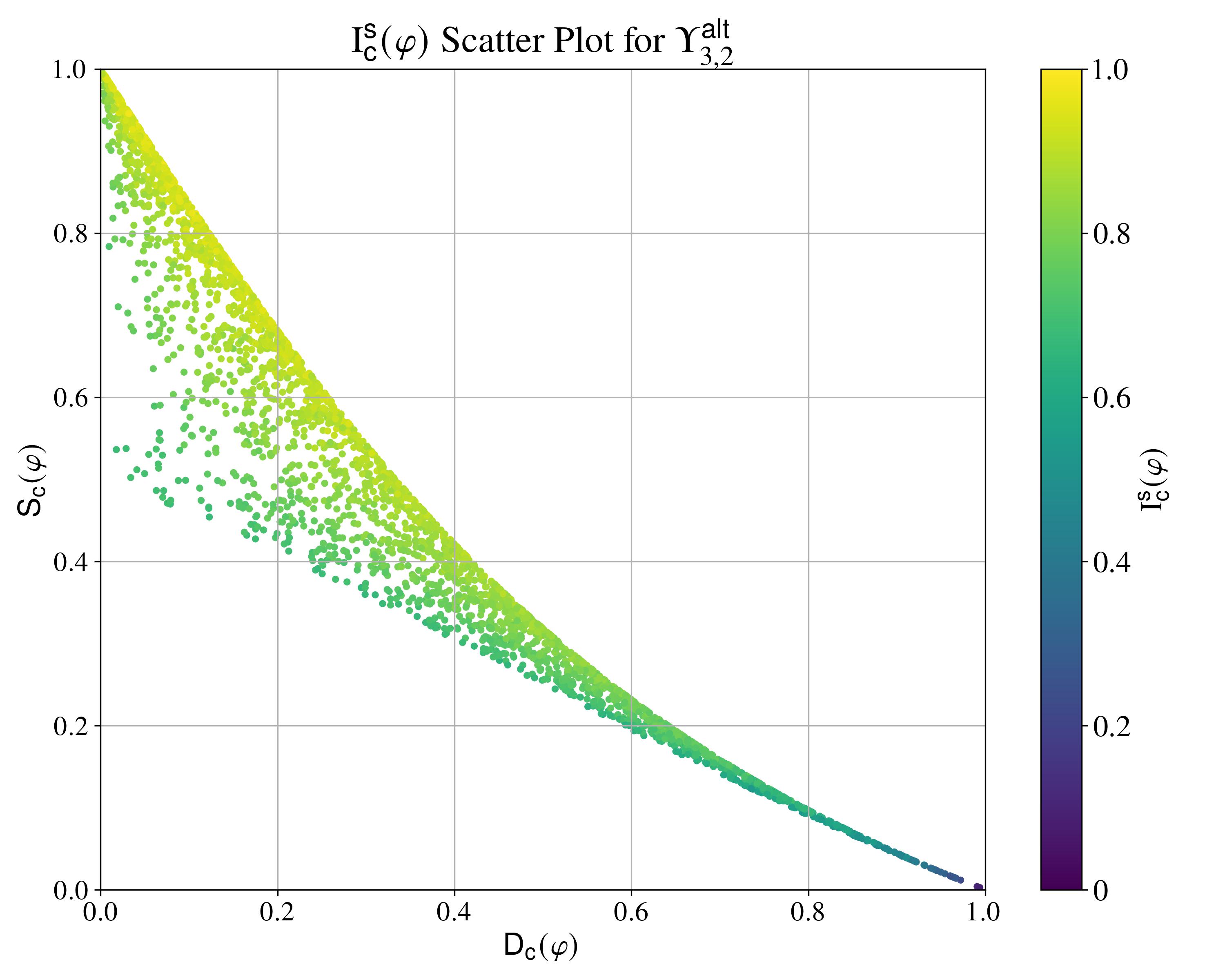}
         \caption{For $\aabs$.} \label{fig:trit-altabs}
\end{subfigure} 
\caption{All $\cabs_{3,2}$ absorbing channels from Figure \ref{fig:trit-cis-3d} plotted for $\skw(\cabs_{3,2})$ and $\cad{\cabs_{3,2}}$. As we go down the values of $\skw(\cchn)$, $|p-q|$ decreases from $0$ to $1$. Thus the top of Figure \ref{fig:trit-stdabs} exists only at $\cfd\cchn=0.5$, since it corresponds to $\{\uabs\}$ (this can also be seen in Figure \ref{fig:ridges}). In a similar vein, note that \textit{all} $\aabs$ occur at $\cfd\cchn=0.5$.} \label{fig:uni-abs-slice}
\end{figure}

\end{widetext}
\end{document}